\documentclass[oneside,english,12pt]{article}
\usepackage[T1]{fontenc}
\usepackage[latin9]{inputenc}
\usepackage{amstext}
\usepackage{graphicx}
\usepackage{amsthm,amsfonts,color}
\usepackage{fullpage}
\usepackage{harvard}
\usepackage{todonotes}
\usepackage{soul}
\usepackage{setspace}

\makeatletter
  \theoremstyle{plain}
  \newtheorem*{lem*}{\protect\lemmaname}
\newtheorem{lemma}{Lemma}
\newtheorem{proposition}{Proposition}
\newtheorem{corollary}{Corollary}
\newtheorem{remark}{Remark}
\newtheorem{theorem}{Theorem}
\newtheorem{definition}{Definition}
\makeatother

\usepackage{babel}
  \providecommand{\lemmaname}{Lemma}

\begin{document}
\title{Optimal Exploration of an Exhaustible Resource \\ with Stochastic Discoveries}
\author{
    Ivar Ekeland\thanks{CEREMADE, Universite Paris-Dauphine},  
    Wolfram Schlenker\thanks{School of International and Public Affairs and The Earth Institute, Columbia University}, 
    Peter Tankov\thanks{ENSAE, Institut Polytechnique de Paris}, and
    Brian Wright\thanks{Department of Agricultural and Resource Economics, University of California at Berkeley}
}

\date{}

\maketitle

\begin{abstract}
\noindent The standard Hotelling model assumes that the stock of an exhaustible resource is known.  We expand on the model by Arrow and Chang that introduced stochastic discoveries and for the first time completely solve such a model using impulse control. The model has two state variables: the ``proven'' reserves as well as a finite unexplored area available for exploration with constant marginal cost, resulting in a Poisson process of new discoveries.  We prove  that  a frontier of critical levels of ``proven'' reserves exists, above which exploration is stopped, and below which it happens at infinite speed. This frontier is increasing in the explored area, and higher ``proven'' reserve levels along this critical threshold are indicative of more scarcity, not less.  In this stochastic generalization of Hotelling's rule, the expected shadow price of reserves  rises at the rate of interest across exploratory episodes. However, the actual trajectories of prices realized prior to exhaustion of the exploratory area {may jump  up or down upon exploration. Conditional on non-exhaustion, expected price arises at a rate bounded above by the rate of interest, consistent with  most  empirical tests based on observed price histories.}


\end{abstract}
\vspace*{2cm}

Key words: Exhaustible resource, optimal exploration, Hotelling's rule, impulse control 

JEL classification: Q32, C61

\newpage
\onehalfspacing
In the seminal work of  \citeasnoun{Hotelling1931},  the  socially optimal  price of a non-renewable commodity with known total reserves, net of marginal extraction cost, rises at the rate of interest. This simple model, in which rising price and falling reserves are both signals of increasing scarcity, has been the core of the economics of nonrenewable resources for nine decades. However, empirical tests of this deterministic model using observed price series are at best mixed \cite{Young1985,HalvorsenSmith1991,ChermakPatrick2001,SladeThille2009}.

Most nonrenewable commodities (for example, petroleum, natural gas, and many other minerals) are characterized by uncertainty about future discoveries. A typical market for such nonrenewable commodities includes, besides known reserves awaiting extraction (as in \citename{Hotelling1931} 1931), an area of what we shall call exploratory resources available for exploration for currently unknown quantities of additional reserves. This can be a crucial feature of the market. For example, the global reserve to production ratio increased from 30 to 50 between 1980 and 2019 (BP Statistical Review of Energy 2021). A model that assumes no growth in total consumable reserves does not appropriately capture a key component of such markets. Moving to a model with both proven reserves and finite explorable resources, one interesting question is whether reserves and observed  prices behave over time as implied by this extended model. 

\citeasnoun{ArrowChang1982} proposed a model that adds exploration with stochastic discoveries to his deterministic ``cake-eating'' model assuming zero cost of extraction of reserves for consumption. The authors call the known commodity stock ``reserves $R$,'' and the unexplored area that might yield further discoveries ``unexplored resources $x$.''   Additional reserves can be identified (``proven'') by exploration of the resource at cost per unit of explored area, resulting in additional discoveries of a known deterministic size that are Poisson distributed. \citeasnoun{ArrowChang1982} conclude that ``The price history will show fluctuations with little upward trend when [the unexplored area] is large; presumably the upward trend is stronger as [the unexplored area] approaches zero, but this requires a probabilistic analysis not yet performed.''  

Our first contribution is to provide the complete and rigorous solution of the Arrow and Chang model using the mathematical theory of impulse control. While Arrow and Chang, as well as several follow-up papers, {start with a bounded exploration rate and use commonsense arguments to conjecture that the rate is infinite at the optimum, we rigorously prove the optimality of impulse strategies within a general class of rules which may include both continuous and jump-type exploration, possibly with simultaneous consumption. This involves substantial mathematical difficulties, yet provides important implications for the realized price of reserves and extraction profile that has puzzled earlier studies. Confirming another conjecture of Arrow and Chang, we prove the existence of a \emph{critical frontier} separating exploration and consumption regions, and show that this frontier is increasing in the amount of explored area. We derive the implications of the model for the behavior of the histories of prices, exploration and reserves given the continuing presence of unexplored resources, showing that while the expected price always rises continuously at the rate of interest, a realized price trajectory may in general jump up or down upon exploration, with a final upward jump upon resource exhaustion.  

Our proof requires several highly non-trivial and non-standard mathematical steps. A defining feature of the model from the mathematical point of view is that information is acquired through exploration, {the timing and  pace of which are determined by the agent.} In other words, unlike standard stochastic control problems where the strategy of the agent is adapted to an exogenous information filtration, here the information filtration depends on the strategy. This makes defining admissible strategies and finding the optimal ones a highly nontrivial task. For this reason, we start by defining a class of "simplified" strategies (which we call bang-bang strategies), where exploration happens in zero time and cannot occur simultaneously with consumption. We use these bang-bang strategies to define the value function and prove rigorously that this value function satisfies an HJB equation (more precisely, a system of inequalities), which is different from the one found in the literature, but is more amenable to mathematical analysis, in particular, it does not require the value function to be differentiable in the explored area. 
We then prove the existence of the critical frontier, show that the frontier is increasing in explored area, and then establish that the value function is differentiable across the frontier. The smoothness of the value function allows in turn to prove that our HJB equation is equivalent to the one found in the literature, and prove that our bang-bang strategies are optimal withing a general class of strategies for this problem. }


Our second contribution is to provide new insights into interpreting empirical tests of the Hotelling model for non-renewable commodities that rely on realized price series. We show that, while price always rises at the rate of interest in expectation, the expected price conditional on not having exhausted all unexplored area rises at less than the rate of interest, explaining why realized historic price series when exploration is still ongoing { often have been found to  rise on average at less than the rate of interest.} Moreover, in the case of markets with continual exploration, which is a defining feature of most exhaustible resources (minerals, fossil fuels, etc), the classical statement that a higher level of ``proven'' reserves implies { higher expected future consumption  and hence less "scarcity" is no longer necessarily true:} as the unexplored area approaches zero, the critical reserve level $R^{\ast}(x)$ when exploration starts increases, as does scarcity. {Finally, we address the Arrow and Chang conjectures regarding the behavior of prices close to exhaustion.}
  
The structure of the paper is as follows. 
Section \ref{model.sec} further discusses and contrasts previous modelling approaches to this problem before introducing the admissible strategies and defining the value function $V\left(x,R\right)$, which we characterize in Section \ref{solving.sec}  as the solution of a suitable HJB equation (Theorem \ref{HJB.thm}). We prove the existence of a free boundary $R^{\ast}\left(x\right)$ separating the exploration region below from the pure consumption region above (Theorem
\ref{cons.thm}) 
and show that $R^{\ast}\left( x\right)  $ increases as unexplored resources $x$ decrease and that the smooth pasting condition holds across the free boundary. 
{As a final result of Section \ref{solving.sec} we consider an alternative class of strategies allowing for both continuous and impulse exploration with simultaneous consumption, and prove that the bang-bang strategies considered in this paper are optimal within this much larger class of strategies, providing the first rigorous proof of optimality of bang-bang strategies in the Arrow and Chang model. }
Section \ref{price.sec} confirms that $\mathbb{E}[p_{t}]=p_{0}e^{rt}$, the expectation of price follows the Hotelling rule. This does not mean that every realization follows the Hotelling rule, even loosely. 
Indeed, our  numerical simulations presented in Section~\ref{numerics.sec} show a wide range of behaviours with minimum reserves increasing, and sawtooth paths of price realizations becoming bounded by the path of price expectations conditional on the initial value of $x$, as the latter approaches zero. Finally, Section~\ref{conclusion.sec} concludes. Details of the mathematical proofs are given in Appendix Section~\ref{proofs.sec}.
\section{Model Setup and the Value Function} \label{model.sec}
\subsection{Previous Literature on Endogenous Discoveries}
Several surveys conclude that empirical tests of the deterministic \citeasnoun{Hotelling1931} model on mineral price histories are at best mixed, even when assumptions regarding reserve homogeneity, extraction cost functions, and market structure are substantially modified \cite{Krautkraemer1998,Kronenberg2008,SladeThille2009,Livernois2009,LivernoisThille2015,CunhaMissemer2020}. 

The one exception where prices have closely followed the Hotelling rule is old-growth timber \cite{LivernoisThilleZhang2006}.  While old-growth timber (Douglas fir) is not strictly non-renewable, the regrowth would at least take decades and the authors argue it is effectively non-renewable. The unique feature of the market is that the stock of old-growth timber can be quite precisely observed with no new ``discoveries.'' 

The distinguishing feature of most other nonrenewable commodities (for example, petroleum, natural gas, and many minerals) is the uncertainty about future discoveries. In a scarcely referenced paragraph in his article, \citeasnoun[p. 174]{Hotelling1931} briefly recognized the need  to extend his dynamic model to make it relevant to nonrenewable commodities in general: 

\begin{quote} The problems of exhaustible resources involve the time in another way besides bringing on exhaustion and higher prices, namely, as bringing increased information, both as to the physical extent and condition of the resource and as to the economic phenomena 	attending its extraction and sale.
\end{quote}

Economists have struggled to construct a stylized stochastic extension of Hotelling's  model to accommodate { costly exploration yielding stochastic discoveries.} \citeasnoun{Slade1982} plotted raw price series for many exhaustible resources and found them to have a U-shaped pattern, although a revised setup by \citeasnoun{BerckRoberts1996} emphasizes that prices are difference and not trend stationary due to stochastic shocks to supply and demand.  \citeasnoun{Gilbert1979} presented a model that included  exploratory resources containing potential reserves, focusing on the informational value of exploration when accumulation of stocks is feasible. { However, Gilbert restricts the total amount of current and potentially discoverable reserves to be equal one of two potential sizes, each with positive probability.} The true size is revealed only when cumulative extraction, at constant marginal cost, equals the smaller size of potential reserves. Exploration is extraction of reserves beyond what is immediately consumed. So if storage cost of reserves is negligible, the cost of exploration is the interest cost of advancing extraction to learn more about the true cumulative size of reserves, avoid a finite probability of a jump of consumption to zero, and more generally  improve the intertemporal allocation of consumption. If the  marginal value of consumption is unbounded, then it is always optimal to maintain via exploration a positive stock of extracted (or ``proven'') reserves. These insights are important for understanding the model we present below. \citeasnoun{SwierzbinskiMendelsohn1989} present a Bayesian model of investment in information about an exhaustible resource  that has a backstop technology that is a perfect substitute. Information shocks shift Hotelling price histories, so mean rates of price increase measured on price histories will differ from the interest rate. The authors argue that this implies that valuations of deposits, which are ex ante expectations, will tend to be more likely to confirm the Hotelling valuation rule than the more error-ridden paths of realizations. 

\citeasnoun{DeshmukhPliska1983} present a controlled process model of  optimal exploration and consumption of a ``non-renewable'' but not exhaustible resource. Exploratory resources are infinite but exploration over time follows a Poisson process with bounded intensity. Utility of zero consumption is zero and marginal utility is bounded; unlike \citeasnoun{Gilbert1979} proven reserves may be zero, and there is no learning; the model is stationary in the sense that consumption, exploration costs  and exploration outcomes are independent of both calendar time  and cumulative discovered reserves. Relevant to results below, they claim (p. 192) to prove that the expected scarcity rent rises at the discount rate as in \citeasnoun{Hotelling1931}. 

{\citeasnoun{ArrowChang1982} presented a mathematically rich model with deep economic implications; it has been studied and extended in several subsequent papers.

\citeasnoun{derzko1981optimal} allow for impulses (infinite exploration rates), assuming that the exploration is instantaneous and continues until a mine is found or the resource is exhausted completely. 
They derive heuristically the variational inequality for the impulse control formulation and, assuming that the utility function is $u(c) = \sqrt{c}$, they provide insights on the solution in the part of the domain when the area to explore is small. In particular, with a first order expansion they show that a frontier separating consumption and exploration regions exists. 
Finally, they argue  that new discoveries always depress the shadow price. {In truth, as we show in this paper, if the discovery comes after only a small decrease in resources, the surprise is positive and the price drops downwards, whereas if a sufficiently large area is explored before a discovery is made, the surprise is negative and the price moves upwards.} 

\citeasnoun{Lasserre1984}, following the same setup as \citeasnoun{ArrowChang1982} and \citeasnoun{derzko1981optimal}, recognizes that the price trajectory may  rise or fall upon exploration, yet asserts that the ``price of reserves is expected to drop upon exploration.'', while we argue in this paper that the price always rises at the rate of interest in expectation. However, the proof in \citeasnoun{Lasserre1984} uses a strict inequality obtained for a discrete change in the unexplored area $X$ (e.g., inequality (2) on page 196). When taking the limit for $\Delta X \rightarrow 0$, this strict inequality reverses to a weak inequality. 

Independently from \citeasnoun{derzko1981optimal}, \citeasnoun{hagan1994optimal} also derive the variational inequality for the impulse control formulation.\footnote{A preliminary version of this paper has circulated since 1981, it is already referenced in \citeasnoun{duffie1986diffusion}.} They prove that the value function is differentiable in the reserves $R$ across the exploration frontier, that the shadow price grows with the rate of interest in the consumption region and that the expected shadow price is continuous during exploration episodes. Furthermore, they use singular expansion techniques to develop intuition about the solution in the small uncertainty regime, that is, when the finds are frequent and of small size. Like  \citeasnoun{derzko1981optimal}, they  claim that new discoveries always depress the shadow price, and that a positive jump in the price may only happen at exhaustion. 

\citeasnoun{duffie1986diffusion} review the contributions of \citeasnoun{hagan1994optimal} and \citeasnoun{derzko1981optimal}, and in particular mention the difficulty of rigourously defining admissible strategies and proving the optimality of impulse control strategies within a general class of strategies. They formalize the problem using the notion of "control time" and propose a diffusion approximation to the original Arrow and Chang model, which changes certain properties of the solution, in particular, simultaneous exploration and consumption become possible in this case. 

\citeasnoun{Quyen1991} extends the Arrow and Chang framework assuming that the distribution of mineral deposits is not known to the agent, who therefore uses Bayesian learning to update the estimate of this distribution as new discoveries are made. The paper has a slightly different setup where the unexplored area is discretized into a finite number of cells. Hence the exploration decision becomes to either explore a cell or not, while in \citeasnoun{ArrowChang1982} and in our approach, the choice of the exploration area is continuous. However, when the number of discrete
cells approaches infinity, in the absence of learning, the two solutions should become identical. Nevertheless, \citeasnoun{Quyen1991} claims that as the exploration program unfolds, the critical level of proven reserves which triggers exploration decreases, whereas we { prove} in this paper that the opposite is true. We believe that the proof in \cite{Quyen1991} has a sign error and the results are hence not due to different modelling assumptions. 

Finally, \citeasnoun{farid1999optimal} apply the theory of piecewise-deterministic Markov processes to the version of the Arrow and Chang model with bounded exploration rate, assuming in addition that the exploration costs are paid from the reserves. They then derive a slightly different HJB equation and { claim that the optimal control is includes  periods of exploration at finite maximum intensity without consumption alternate with periods of consumption without exploration.}


}

\subsection{Model Setup}
We build on the familiar cake-eating problem \cite{Hotelling1931}. There is a single, infinite-lived consumer, who discounts a flow of instantaneous utility at the constant rate of interest, $r>0$, which coincides with the rate of time preference. There is a single, non-renewable, good, initially available as a finite quantity of proven reserves, $R_{0}>0$. Reserves do not depreciate over time and can be stored at no cost. We assume that there is implicitly an unlimited supply of another good -- called ``goods in general'' in \citeasnoun{ArrowChang1982} -- which is the numeraire. { The utility of the consumer is linear in the numeraire,  and  exploration costs, which we shall consider in the following section, are paid in units of numeraire. }

As in Hotelling's simplest case, we assume no cost of extraction of reserves for consumption.\footnote{The {role of reservoir pressure in limiting the response of production from oil wells is} addressed in Anderson, Kellog and Salant (2018). } The rate of consumption of the nonrenewable good at time $t$ is $c(t)\geq 0$,  $R(t)\geq 0$ is the  quantity of proven reserves at time $t$, and the instantaneous utility function gross of any cost of exploration $u:\left(  0,\ \infty\right)  \rightarrow\mathbb{R}$, measured in units of numeraire,
is concave, increasing and $C^{2}$, with\ $u^{\prime\prime}\left(  c\right) <0$ for $c>0$. 

The representative agent's problem  is to choose { the rate of  consumption of the non-renewable good} so as to maximise the present value of consumer surplus  from its consumption  over time. 
The optimization problem is:
\begin{eqnarray*}
\mathcal U(R_0)&=&\max_{c\left(  \cdot\right)  }\int_{0}^{\infty}u\left(  c\left(  t\right)
\right)  e^{-rt}dt \qquad\text{subject to}  \\
 \frac{dR}{dt}  & =  & -c(t),\ \ c\left(  t\right)  \geq0,\ \ R(t)\geq 0\ \ \forall t\geq 0.\\
R\left(  0\right) & =& R_{0}.%
\end{eqnarray*}

Here $\mathcal U(R_0)$ is the value of $R_0$, the present value of the future stream of instantaneous utility the agent draws from consuming an aggregate  quantity $R_0$ of the good.\footnote{Since this present value is measured in units of numeraire, {the  price that the agent is willing to pay for an extra unit of non-renewable good, is given by the derivative of the value function:}
$p_0 = \mathcal U'(R_0),$ and the same holds for all future dates: $p_t = \mathcal U'(R_t).$} Optimization of consumption over time implies Hotelling's rule:
$p_t = p_0 e^{rt}$. In his  deterministic model of optimal consumption of a finite stock of reserves $R$, rising price and falling  reserves are equivalent  signals of increasing scarcity. 
Before we  extend this model to include stochastic exploration, a few further notation:

Throughout the paper, unless specified otherwise, we will take the utility function $u\left(  c\right)  =\frac{1}{\alpha}c^{\alpha}$, with $0<\alpha<1$.\footnote{We denote by $u^*$ the convex conjugate of $u$, given { by $u^*(p) = \sup_{c>0} \{u(c) - cp\} = \frac{1-\alpha}{\alpha}p^{\frac{\alpha}{\alpha-1}}.$}} With our choice of utility function, the present value of future utility in the Hotelling problem without exploration is given explicitly by 
\begin{equation}
\mathcal U(R) =  \left(\frac{1-\alpha}{r}\right)^{1-\alpha} \frac{R^\alpha}{\alpha}.\label{uhotel}
\end{equation}

Following \citeasnoun{ArrowChang1982}, we now introduce a costly exploration process that can yield additional reserves. They assume a known area  of a potentially mineral-bearing resource, identified  for example by its observable surface geological characteristics, containing an unknown number of deposits following a Poisson distribution, each consisting of reserves of size $a>0$. 

Exploration of any interval of unexplored resources { may proceed at a finite rate or instantaneously.} As we establish below, if exploration were costless the entire unit resource would be instantaneously explored at the beginning, generating perfect knowledge of the total stock of reserves. With this information, the path of consumption can be fully optimized, exploiting the costless storability of this stock with the reserve level as the sole state variable, and the Hotelling rule holds. 

Henceforth, we assume a constant marginal cost $k>0$ for exploration of the resource. If exploration starts at $x$ and stops at $x^{\prime}<x$, because a deposit is found at $x^{\prime}$ or $x^{\prime}=0$ (there is nothing left to explore), then the exploration cost incurred is $k\left(x-x^{\prime}\right).$ In the exploratory regime there are two state variables, the remaining unexplored interval $x\geq 0$, and the reserves $R>0$ already discovered (``proven'') that have not yet been consumed, (henceforth called the ``reserves'').  
When $x$ approaches zero, unexplored resources are exhausted and the model transitions via a price jump to the deterministic Hotelling consumption regime starting with current reserves $R$. 

In this dynamic stochastic model,  optimization  of social welfare under uncertainty involves  trade-offs. The value of  early exploratory information regarding remaining reserves for optimization of the subsequent consumption path, given the non-negativity of reserves,  must be balanced against the benefits of delaying investment in costly exploration.
 
Our results are as follows. The agent's strategy is completely determined by the critical reserve frontier  which is a smooth decreasing function of unexplored resources, $R^*:[0,\infty)\to (0,\infty)$.  
When reserves are above $R^*$, $x$ is constant and the agent is in consumption mode; an instantaneous exploratory episode occurs when reserves decline to $R^*$. If exploration finds  one or more deposits that raise reserves above the critical level given remaining resources, exploration stops until consumption reduces reserves to the  critical level. 

For given unexplored area $x$, there is a one-to-one correspondence between reserves and price. 
Hence there is also a  critical price level $P^*:[0,\infty)\to (0,\infty)$. 
It is  optimal to explore  only when price is not below this critical level.
 
In finite time, after the last exploratory episode the resources will decline to zero.  At this point, the exploratory regime terminates and price jumps up to the deterministic increasing path that follows the familiar Hotelling rule, given proven reserves at the time of transition, which subsequently decline monotonically.  

A regime-ending jump occurs with positive probability in any exploratory episode. This fact has important implications for price behavior within the exploratory regime defined by $x>0$. In this regime, the price path rises at the rate of interest between exploratory episodes. Since the expected price path also rises at the rate of interest and price jumps upward if the regime ends in exhaustion of unexplored resources at that episode, it must jump down in expectation, conditional on remaining within the exploratory regime. Given positive  resources remaining after exploration, price again rises smoothly at the discount rate, until the next exploratory episode. In general, realized price in the exploratory regime has this saw-tooth pattern. 
 
After a sufficiently large interval of unsuccessful exploration, the increase in the exploration frontier as $x$ decreases dictates that realized price might have a positive jump even though more reserves are discovered and exploration ceases with positive remaining unexplored resources. Hence a price path within the exploratory region including such positive jump might have an average rate of price increase greater than the rate of interest. Furthermore, the common discrete size of each deposit means there is an $\hat{x}>0$ such that starting at unexplored resources $x<\hat x$ all price paths realized within the exploratory regime are bounded above by a path rising at the rate of interest. 
Taken together, these two results falsify the conjecture of Arrow and Chang that the upward trend in realized price  will be strongest as unexplored resources approach exhaustion.
 
 The realized price is a valid measure of the scarcity of the mineral; it takes into account anticipated  future exploratory outcomes given reserves and unexplored resources. On the other hand, a positive jump in reserves does not necessarily imply reduced scarcity.

\subsection{Illustration of the Model: Sample Path}\label{samplePath.sec}
Before we dive into the mathematical model, an analysis of the optimal consumption-exploration strategy might offer further insights into the model. A typical path is shown in Figure~\ref{samplePath.fig}.  The left graph illustrates the optimal exploration policy by plotting known reserves (y-axis) against the unexplored area (x-axis). There exists a critical reserve level $R^{\ast}(x)$ shown in yellow that decreases in the explored area $x$.  If known reserves exceed the critical level, exploration is zero and one is the consumption phase analogous to the standard Hotelling model: ``proven'' reserves are drawn down and price rises at the rate of interest. As soon as consumption decreases the stock of known reserve level to the critical level, exploration starts at infinite speed (zero time) until either enough new discoveries are found to push the stock of known reserves above the critical level or all unexplored area is exhausted. A sample path of various quantities described by our model is shown by $A \rightarrow B_0 \rightarrow B \rightarrow C \rightarrow D_0 \rightarrow D \rightarrow E \rightarrow F_0 \rightarrow F \rightarrow G_0 \rightarrow G \rightarrow H \rightarrow K \rightarrow L$. The right graph plots unexplored area, known reserves, price and consumption rate (y-axes) against time (x-axis), using the same letters to mark events.  Note that the points with ``zero'' subscript do not appear in the right graph because they are indistinguishable from the points without the subscript since the exploration happens in zero time.  

The sample path starts at point $A$ when known reserves are large enough so exploration is zero, price rises at the rate of interest, and the consumption rate decreases in time until known reserves hit level $B_0$, the critical reserve level when the system switches from consumption to exploration. The unexplored area decreases in zero time until a new discovery is made ($B$), which pushes known reserves up by the normalized size of the discovery ($C$). Note on the right graph how all variables jump in zero time to their new levels between $B$ and $C$.  The size of the jumps depends on the random amount of explored area that is required for the next discovery. At point $C$, reserves are again above the critical reserve level and the process of no exploration, price rising at the rate of interest, consumption rate declining repeats itself $C \rightarrow D_0$, before the next exploration starts when proven reserves are drawn down to $D_0$.  The next cycle $D_0 \rightarrow D \rightarrow E$ follows the analogous pattern of exploration and subsequent consumption $E \rightarrow F_0$. The next cycle is slightly different: the required area to find the next reserve ($F_0 \rightarrow F$ in the left graph) is large, so even though the new discovery pushes the stock of known reserves up to $G_0$, it is still below the critical level. As a result, the next exploration starts immediately ($G_0 \rightarrow G$), and the discovery pushes the known reserves above the critical level again ($G \rightarrow H$).  There are two noteworthy observations: first, since the second exploration starts right away in zero time, the right graph shows the system jumping from $F \rightarrow H$ in zero time, sides-stepping point $G$. Second, the price increases as we move from $F$ to $H$ ; the unexplored area decreased so much that despite the two new discoveries and higher reserve level, scarcity increased. Lastly, the next cycle starts when consumption decreases the reserve stock from $H$ to $K$, which results in no additional discoveries and the unexplored area $x$ is exhausted ($L$), at which point prices jump up one last time before the follow the classical Hotelling path of rising at the rate of interest. Prices will always show a positive jump when the unexplored area is exhausted.

\subsection{Admissible strategies and the Value Function}\label{value.sec}
In this section we define the value function of our optimal consumption-exploration problem. Let $N$ denote a Poisson process with intensity $\lambda$, which models our stochastic exploration process. The jump times of $N$, which correspond to the locations of the deposits, will be denoted by $(\xi_n)_{n\geq 0}$, with $\xi_0 = 0$. Denote by $(\mathcal F_u)_{u\geq 0}$ the natural filtration of the process $N$.

A {bang-bang} consumption-exploration strategy consists of an
increasing sequence of random variables $(\theta_n)_{n\geq 0}$ (moments where the exploration starts and continues either until the next find or until the explorable land is exhausted)  and a sequence of random maps $c_n : \Omega\times\mathbb R_+ \to \mathbb R_+$, where $c^0$ defines the consumption strategy before the first exploration date, and for $n\geq1$, $c^n$ defines the consumption strategy in the interval between $\theta_{n-1}$ and $\theta_n$. {We call such strategies bang-bang since within such a strategy, exploration occurs in zero time until either the next deposit is found or the entire resource is exhausted, and no exploration takes place during a consumption episode. In this paper, we first study the  strategies of bang-bang type. In Theorem \ref{strat.thm} we characterize the optimal strategies from this class. Then, in Theorem \ref{optimality}, we show that such strategies are optimal within a general class of strategies, where exploration may take place with a finite rate and may be simultaneous with consumption.}

%


To define which strategies are admissible, assume that the agent starts at time $t=0$ with reserve level $R$ and unexplored area $x$. We call the pair $(x,R)$ initial data of the problem. 
Define the explored after $n$ exploration episodes, $X_n$ as follows:
$$
X_{n} = \xi_n\wedge x \quad n\geq 0. 
$$
The total number of exploration episodes is $N_{x}+1$ and the consumption process is given by 
\begin{equation}
c_t =  c^0(t)\mathbf 1_{t<\theta_0} + \sum_{n= 1}^{N_x} c^n(t) \mathbf
1_{\theta_{n-1}\leq t< \theta_{n}} + c^{N_{x}+1} \mathbf
1_{\theta_{N_x}\leq t}.\label{state.cons}
\end{equation}

The $\sigma$-field $\mathcal F_{X_n}$ contains the information available to the agent after $n$ exploration episodes. Recall that a set $A$ belongs to $\mathcal F_{X_n}$ if and only if, for all $u\geq 0$, $A\cap \{X_n \leq u\} \in \mathcal F_u$. 

\begin{definition}[Admissible bang-bang strategy]\label{defbb}We say that the consumption-exploration strategy $(\theta,c)$ is admissible for initial data $(x,R)$ if
\begin{itemize}
\item[i.] For each $n\geq0$, the random variable $(\theta_n)$ is measurable with respect to $\sigma$-field $\mathcal F_{X_n}$,
and the function $c^n$ is $\mathcal F_{X_n} \times\mathcal B(\mathbb R_+)$ measurable; 
\item[ii.] The budget constraint
\begin{equation}
  R_t:=R - \int_0^t c_s ds + a \sum_{n=0}^{N_x-1} \mathbf 1_{\theta_n
    \leq t} \geq 0\label{budget}
\end{equation}
is satisfied a.s. for all $t\geq 0$. 
\end{itemize}

\end{definition}

Condition i.~ensures that the strategy of the agent at time $t$ may only depend on the information acquired at time $t$ through exploration, and equation (\ref{budget}) guarantees that the reserve level remains positive at all times. We denote the set of all such admissible strategies by $\mathcal A(x,R)$.

The value function is defined by
\begin{equation}
V(x,R) = \sup_{(\theta,c)\in \mathcal A(x,R)} \mathbb E\left[\int_0^\infty e^{-rt} u(c_t)
  dt- k\sum_{n=0}^{N_x} e^{-r\theta_n}(X_{n+1}-X_n)\right]\label{vf}
\end{equation}
for $(x,R)\in  [0,\infty)^2$. 

We next present a proposition gathering some a priori results regarding the behavior of the value function $V$. 
\begin{proposition}${}$\label{apriori.prop}
\begin{itemize}
    \item[i.] If the exploration is costless ($k=0$), it is optimal to explore all unexplored area immediately at $t=0$, and the value function is given by
$$
V(x,R) = \sup_{(c,\theta)\in \mathcal A(x,R)} \mathbb E\left[\int_0^\infty e^{-rt} u(c_t)
  dt\right] = \mathbb E[\mathcal U(R+aN_{x})].
$$
\item[ii.] The value function is increasing and locally Lipschitz continuous in the reserve level $R$.
\item[iii.] The value function admits the following bounds
\begin{equation}
\mathcal U(R) \leq V(x,R)\leq \mathbb E[\mathcal U(R + aN_{x})]\leq C(1+R^\alpha),\quad R\geq 0,\label{ubound}\label{lbound}
\end{equation}
for some $C<\infty$.
\end{itemize}
\end{proposition}

To prepare the ground for solving the model in the next section, we define the following \emph{exploration operator}, applied to the value function: 
\begin{equation}
MV(x,R)  = \int_0^{x} V(x-s,R+a) \lambda e^{-\lambda s} ds + \mathcal U(R)
e^{-\lambda x} - k \frac{1-e^{-\lambda x}}{\lambda}.\label{exp.op}
\end{equation}
This operator represents the value of starting exploration immediately and exploring until either a new deposit is found or the entire area is exhausted, when the reserve level equals $R$ and the unexplored area equals $x$.

In particular, when the reserve level is zero, from equation (\ref{lbound}),
\begin{eqnarray*}
MV(x,0) &=& \int_0^{x} V(x-s,a) \lambda e^{-\lambda s} ds - k \frac{1-e^{-\lambda x}}{\lambda} \\ &\geq& \int_0^x \mathcal U(a) \lambda e^{-\lambda s} ds - k \frac{1-e^{-\lambda x}}{\lambda}  = (1-e^{-\lambda x})\left(\mathcal U(a) - \frac{k}{\lambda}\right).
\end{eqnarray*}
We thus arrive to a natural and simple condition
\begin{equation}
\mathcal U(a) \geq \frac{k}{\lambda},\label{assump2}
\end{equation}
which guarantees that exploration is optimal at zero reserve level for all $x> 0$. Throughout the paper we shall always assume that this condition is satisfied and introduce the parameter
\begin{equation}
\varepsilon:= \frac{k}{\lambda \mathcal U(a)}. \label{epsdef}
\end{equation}

\section{Solving the Model}\label{solving.sec}
We now show that the value function satisfies a variant of the HJB equation.

\begin{theorem}[Characterization of the value function]\label{HJB.thm}
  The value function $V(x,R)$ is concave and continuously differentiable in $R$ and increasing and continuously differentiable in $x$ on its entire domain. It is the solution of the following HJB equation: 
$$
\max\left\{u^*\left(\frac{\partial V}{\partial R}\right) -
    rV,MV-V\right\}=0,\quad V(0,R) = \mathcal U(R).
  $$
Conversely, if a non-negative function $\widetilde V(x,R)$ is continuously differentiable in $R$ on $(0,\infty)$ for every $x\geq 0$, satisfies the HJB equation, and admits the bound
$$
\widetilde V(x,R) \leq C(1+U(R)),\quad R\geq 0,
$$
for some $C<\infty$,  it is given by equation~(\ref{vf}).
\end{theorem}

\begin{remark}
{ As we pointed out in the introduction, the above characterization of the value function differs from the HJB inequalities obtained by other authors in the context of the Arrow and Chang model, given below and proved in Lemma \ref{prop.lm} in the Appendix.
\begin{equation}
\max\left\{u^*\left(\frac{\partial V}{\partial R}\right)-rV,\lambda(V(x,R+a)-V(x,R)) - \frac{\partial V}{\partial x} - k\right\} = 0,\quad V(0,R) = \mathcal U(R). \label{hjb.classic}
\end{equation}
While formally the two expressions are equivalent, on the one hand, our formulation is more natural in the context of impulse control, since the value of no exploration is compared directly to the value obtained after one complete exploratory episode rather than the value of exploring an infinitesimal amount of land, and on the other hand, our setting is more amenable to mathematical analysis, as it does not require the value function to be differentiable in $x$, which is difficult to prove a priori.}
\end{remark}


From the HJB relation above, it follows that the domain $(x,R)\in  \mathbb R^2_+$ can be divided into two disjoint domains: $\mathcal C = \{(x,R):V(x,R)>MV(x,R)\}$ and $\mathcal E = \{(x,R):V(x,R)=MV(x,R)\}$. In the domain $\mathcal C$, starting the exploration immediately decreases the value function of the agent, it is therefore not optimal to explore, and the agent will consume until the reserve level becomes sufficiently small to start the exploration. We call $\mathcal C$ the \emph{consumption region}. In the consumption region, the value function satisfies the equation
$$
u^*\left(\frac{\partial V}{\partial R}\right) =rV.
$$
This equation is satisfied by the Hotelling value function $\mathcal U$ over the entire domain. 
By contrast, in the region $\mathcal E$, starting the exploration immediately does not decrease the value function, so we call this region the \emph{exploration region}. 

The following theorem provides a precise characterization of the consumption and the exploration regions. 
\begin{theorem}[Characterization of the consumption and exploration regions]\label{cons.thm}
For every $x>0$ there exists $R^*(x)\in (0,\infty)$ such that $\mathcal C = \{(x,R): R> R^*(x)\}$ and $\mathcal E = \{(x,R): R\leq R^*(x)\}$. The function $R^*(x)$ is decreasing, continuously differentiable, and satisfies
$ R^*(0) = R_0$
where $R_0$ is the solution of 
$$
\alpha\left(1+\frac{a}{R_0}\right)^{\alpha-1}  +  (1-\alpha)\left(1+\frac{a}{R_0}\right)^\alpha-(1-\alpha)\left(\frac{a}{ R_0}\right)^{\alpha}\varepsilon  = 1.
$$
\end{theorem}
The above theorem establishes the existence of a \emph{critical reserve frontier}: a smooth function $R^*(x)$ such that it is optimal to explore when and only when the reserves fall below $R^*(x)$. The critical reserve frontier is decreasing in $x$ (equivalently, increasing in the amount of explored area): with increasing scarcity it becomes beneficial to explore early, as this allows the agent to optimize the consumption of the remaining reserves.

The understanding of the shape of consumption and exploration regions enables us to give a precise characterization of the optimal {bang-bang} strategy in the next theorem. 
\begin{theorem}[Characterization of the optimal bang-bang strategy]\label{strat.thm}
  Starting with initial data $(x,R)$, the optimal bang-bang consumption-exploration strategy is defined as follows. Before the first exploration date,
   \begin{eqnarray}
c^0(t) &= & c_0 e^{-\frac{r}{1-\alpha}t},\quad c_0 = \left(\frac{\partial V}{\partial R}(x, R)\right)^{\frac{1}{\alpha-1}},\label{strat.c0}\\
    \theta_0 &= & \inf\{t\geq 0: R^0(t)\leq R^*(x)\}, \label{strat.theta0} 
\\
    R^0(t)& =& R - \frac{1-\alpha}{r} c_0 \left\{1-e^{-\frac{r}{1-\alpha}t}\right\}. \label{strat.R0}
   \end{eqnarray}
After the first exploration date,  for $n\geq 1$, 
  \begin{eqnarray}
c^n(t) &= & c_n e^{-\frac{r}{1-\alpha}(t-\theta_{n-1})},\quad c_n = \left(\frac{\partial V}{\partial R}(x-X_n, R_{\theta_{n-1}})\right)^{\frac{1}{\alpha-1}},\label{strat.c}\\
    \theta_n &= & \inf\{t\geq \theta_{n-1}: R^n(t)\leq R^*(x-X_n)\}, \label{strat.theta} 
\\
    R^n(t)& =& R_{\theta_{n-1}} - \frac{1-\alpha}{r} c_n \left\{1-e^{-\frac{r}{1-\alpha}(t-\theta_{n-1})}\right\}, \label{strat.R}\\
    R_{\theta_{n-1}} &=& R^{n-1}(\theta_{n-1})+ a\mathbf 1_{\xi_{n}\leq x}. \label{strat.Rn}          
  \end{eqnarray}
\end{theorem}
The agent consumes at an exponentially decaying rate until the reserve level drops below the critical reserve frontier. At this time, exploration starts and continues until either the reserve level is brought above the critical reserve frontier through one or several successful exploratory events, or the entire remaining area is explored. After the exploratory episode, the consumption resumes.

Up to now, we considered \emph{bang-bang} strategies, where the agent consumes until hitting the exploration frontier and then explores a finite area in zero time, until either a new deposit is found or the whole interval is explored. We finally prove that these strategies are optimal within a much larger class of strategies allowing for simultaneous exploration and consumption, defined below. 

\begin{definition}\label{defadm2}
We say that a couple of stochastic processes $(c,X)$, where $c_t$ stands for the consumption rate at time $t$, and $X_t$ stands for the area explored until time $t$, is an admissible  consumption-exploration strategy for initial data  $(x,R)$ if
\begin{itemize}
    \item The process $c$ takes values in $\mathbb R_+$, and the process $X$ is increasing and satisfies $X_0 = 0$ and $X_t \leq x$ for all $t\geq 0$. 
    \item The processes $X^{-1}$ and $u\mapsto \int_0^{X^{-1}(u)}c_s ds $ are $(\mathcal F_u)$-adapted, where $X^{-1}$ denotes the right-continuous inverse mapping of $t\mapsto X_t$, defined by $X^{-1}(u) = \inf\{t: X(t)>u\}$.
    \item The budget constraint 
    $$
    R - \int_0^t c_s ds + a N_{X_t} \geq 0
    $$
    is satisfied for all $t\geq 0$. 
\end{itemize}
\end{definition}

An admissible bang-bang strategy of Definition \ref{defbb} is characterized by the consumption process given in Equation (\ref{state.cons}) and explored area process
\begin{equation}
X_t = \sum_{n=0}^{N_x-1} X_{n+1} \mathbf 1_{[\theta_n,\theta_{n+1})}(t) + x \mathbf 1_{t\geq \theta_{N_x}}.\label{area.eq}
\end{equation}
It is easy to see that this strategy is also admissible according to the Definition \ref{defadm2}. 
The following theorem shows that no admissible strategy may procure a value larger than the value function defined with bang-bang strategies. The bang-bang strategy of Theorem \ref{strat.thm} is therefore optimal among all admissible strategies. 
\begin{theorem}\label{optimality}
Let $(c,X)$ be an admissible consumption-exploration strategy for initial data $(x,R)$. Then, 
$$
 \mathbb E\left[\int_0^\infty e^{-rt} u(c_t) dt - k\int_0^\infty e^{-rt} dX_t\right] \leq V(x,R). 
$$

\end{theorem}



\section{Price trajectories and Hotelling in expectation}
\label{price.sec}
{Following the Hotelling framework, the \emph{shadow price} of reserves as function of unexplored area $x$ and reserve level $R$ in our model may  be defined as follows:
$$
p(x,R) = \frac{\partial  V}{\partial
    R}(x,R).
  $$
From the concavity of the value function, it follows that the price is decreasing as function of the reserve level $R$. On the other hand, in the exploration region, the price is increasing as function of unexplored area $x$. 
In this region, $V(x,R) = MV(x,R)$ and from the proof Theorem \ref{HJB.thm}, first part, it follows that
$$    p(x,R) = \int_0^{x} p(s,R+a) \lambda e^{-\lambda (x-s)}ds + \mathcal U'(R) e^{-\lambda x}.
$$
Differentiating both sides, we obtain the following representation for the derivative of the price function in the exploration region: 
$$\frac{\partial}{\partial x} p(x,R) = \lambda (p(x,R+a) - p(x,R))\geq 0,
$$
which means that  unsuccessful exploration (with constant $R$) leads to increased scarcity and pushes the price of reserves upwards. 
}    

{Given the optimal consumption-exploration strategy $(c,\theta)$ as described in Theorem \ref{strat.thm} the associated explored area process $(X_t)$ (defined by Equation (\ref{area.eq})) and reserve process $(R_t)$, we may define the \emph{price process} $p_t = p(x-X_t,R_t)$.}
  During a consumption period, the price process satisfies
  $$
  p_t = p_{\theta_{n_t}}e^{r(t-\theta_{n_t})},
  $$
  where $n_t = \max\{n:\theta_n\leq t\}$. In other words, the price grows according to the Hotelling rule until the next exploratory episode, when it can jump either upwards or downwards. {During the exploratory episode, there is an initial price increase (because $p(x,R)$ is decreasing in $x$ in the exploration region). If a new deposit is found, this initial increase is followed by a downward jump, since the exploration frontier $R^*(x)$ is decreasing in $x$, and the price $p(x,R)$ is decreasing in $R$.  When no deposit is found, there is no downward jump, which means that upon resource exhaustion the price always jumps upwards.}  Since $\frac{\partial V}{\partial R}$ is decreasing in $R$, one can argue that the consumption strategy consists in consuming until the price $p_t$ hits the critical price given by $\frac{\partial V}{\partial R}(x_t,R^*(x_t))$. 
  

Since the price jumps upwards or downwards at any exploratory episode, the Hotelling rule cannot hold for a specific price trajectory. However, the following result shows that in expectation, the Hotelling rule still holds. This shows that on average, downward jumps of the price process, which often occur in the beginning of exploration are compensated by the upward jumps, which usually occur at or near the end. 
\begin{proposition}
The discounted price process $p_t e^{-rt}$ is a martingale and in particular for all $t\geq 0$, we have $\mathbb E[p_t] = p_0e^{rt}$. 
\end{proposition}
\begin{proof}
The price process is continuous within the consumption region and jumps immediately when $R_t$ hits the boundary of the exploration region. Inside the consumption region the value function is smooth and its second derivative satisfies 
$$
\frac{\partial^2 V}{\partial R^2}(x-X_t,R_t) =-r\frac{\partial V}{\partial R} (x-X_t,R_t)I\left(\frac{\partial V}{\partial R} (x-X_t,R_t)\right)^{-1} = -rp_t c_t^{-1}
$$
The price process can then be written as follows.
\begin{eqnarray*}
e^{-rT}p_T &= & p_0 - \int_0^T rp_t e^{rt} dt+ \int_0^T e^{-rt}\frac{\partial^2 V}{\partial R^2}(x-X_t,R_t)  dR_t \\ 
 & & \qquad \qquad +\sum_{i=0}^{\infty}\mathbf 1_{\theta_i \leq T}\, e^{-r\theta_i}\{\frac{\partial V}{\partial R} (x-X_{\theta_i},R_{\theta_{i}}) -\frac{\partial V}{\partial R} (x-X_{\theta_i-},R_{\theta_{i}-}) \}\\
& = & p_0  +\sum_{i=0}^{\infty}\mathbf 1_{\theta_i \leq T}\, e^{-r\theta_i}\{\frac{\partial V}{\partial R} (x-X_{\theta_i},R_{\theta_{i}}) -\frac{\partial MV}{\partial R} (x-X_{\theta_i-},R_{\theta_{i}-}) \},
\end{eqnarray*}
where we have used the smooth pasting principle in the last line. From the proof Theorem \ref{HJB.thm}, first part, it follows that
$$
\frac{\partial MV}{\partial R}(x,R) = \int_0^{x} \lambda e^{-\lambda h}\frac{\partial V(x-h,R+a)}{\partial R} dh + e^{-\lambda x} U'(R).
$$
Therefore,
\begin{eqnarray*}
  \mathbb E\left[\frac{\partial V}{\partial R} (x-X_{\theta_i},R_{\theta_{i}}) \Big|\mathcal F_{\xi_i}\right] &=&\mathbb E\left[\frac{\partial V}{\partial R} ((x-\xi_{i+1}) \vee 0,R_{\theta_{i}-}+a \mathbf 1_{x- \xi_{i+1}>0}) \Big|\mathcal F_{\xi_{i}}\right]  \\ &=&\frac{\partial MV}{\partial R} (x-X_{\theta_i-},R_{\theta_{i}-}), 
\end{eqnarray*}
since $X_{\theta_i} = \xi_{i+1}$. We conclude by the law of iterated expectations that $e^{rt}p_t$ is a martingale. 
\end{proof}


\section{Numerical simulations}\label{numerics.sec}
In this section we present several numerical illustrations of our analytical solution. We use two sets of parameters corresponding to low and high intensity of discoveries. In the first case, the discovered amount is $a = 2.5$, the discovery intensity is $\lambda = 2$ (meaning that there is, on average, only two mines in the whole interval), and the exploration cost $k = 5$. In the second case, the find size is $a=0.5$, the find intensity $\lambda=10$ and the exploration cost $k=1$. This choice highlights the crucial role of the form of uncertainty by picking parameters where the expected number of discoveries $a \times \lambda = 5$ and the cost per unit of discovery $\frac{k}{\lambda} = 0.1$ are the same in both cases.  The only difference is that the uncertainty (standard deviation) of discoveries $a \sqrt{\lambda} = \frac{5}{\sqrt{\lambda}}$ is decreasing in $\lambda$. The utility parameter is $\alpha=0.5$ for both examples and the interest rate is $r=0.02$. 

Figure \ref{frontiers.fig} shows the critical level that separates the exploration and consumption regions (left graph) and the price at the critical reserve level (right graph) for the two parameter sets. The green line corresponds to the low intensity of discoveries, and the red line to high intensity of discoveries. The left graph displays the critical reserve stock when exploration starts as well as the exploration region (shaded area below the critical reserve level). Optimal exploration follows a bang-bang solution: it is zero if reserves exceed the critical reserve level (consumption region), but start at infinite speed as soon as consumption decreases proven reserves to the critical level. Exploration stops either if new discoveries bring the known reserves again above the critical reserve level or if all unexplored area is exhausted. The right graph displays the price of the resource at the critical level. Recall that expected total discoveries are the same in both cases, yet the critical reserve stock is different. In case of lower uncertainty (red line), the critical level when exploration starts is lower when the explored area is small, as it is almost impossible that no discoveries are made and it hence is preferable to defer exploration cost. However, as the explored area becomes large, the critical level actually becomes larger relative to the case of higher uncertainty.


Lastly, Figure  \ref{price10.fig} plots the average evolution and the quantiles of various quantities computed over 1000 simulations (in the second case).  Black lines display the average resource price, explored area, reserve level and the consumption rate over time. Purple lines display the averages taken only over the trajectories which have not reached exhaustion at a given time. The right part of purple trajectories is noisy because most of simulations have reached exhaustion by this time. Shaded areas display the distribution of outcomes. The price graph in the top left also displays a path that rises at the rate of interest in red, which equals the average price path.  However, the price average taken only over the trajectories which have not reached exhaustion rises at a much slower rate. 



\section{Conclusion}\label{conclusion.sec}
New stochastic discoveries are an important challenge for optimal  exploration and consumption strategies  for non-renewable commodities when the quantity of reserves hidden in unexplored resources is unknown. We are the first to fully and rigorously solve the problem of a model with these realistic features. If exploration is instantaneous and costless, the solution is to explore the entire resource immediately, and then follow Hotelling' deterministic rule, given total proven reserves.    But if exploration is costly, optimal timing of exploratory episodes    balances the value of getting more precise  information on the total  amount of reserves to be consumed against the present value of  delaying exploration cost into the future. 

We find that the optimal policy follows a {
on-off strategy, commuting between two regimes}. Once the known reserve stock falls to a critical level that increases as the  remaining unexplored area declines, exploration starts at infinite speed until either the proven reserve stock again exceeds the critical level or the entire remaining unexplored area is exhausted. We show that the critical reserve level is decreasing in the unexplored area. If proven reserves are above the critical level through new discoveries that follow a Poisson process, exploration stops and price follows a classical price path that rises at the rate of interest.

The paper provides several new insights into tests of Hotelling's rule on  price histories of non-renewable commodities resources, many of which have not risen as predicted in the simple deterministic Hotelling model. We show that while the price path always rises at the rate of interest in expectation, the expected  price path conditional on not having run out of unexplored area rises  at less than the rate of interest. This is explains why forward-looking tests of the Hotelling rule generally do not reject it \cite{MillerUpton1985}, while backward looking tests do \cite{HalvorsenSmith1991}. Starting below a positive threshold level  of unexplored resources, all price paths in the exploratory regime rise at less than the rate of interest. 
Moreover, since the critical  reserve stock is increasing as the unexplored area decreases,  a higher reserves level does not necessarily indicate lower scarcity. 


\ifx\undefined\bysame
\newcommand{\bysame}{\leavevmode\hbox to\leftmargin{\hrulefill\,\,}}
\fi

\clearpage
\begin{figure}
\caption{Sample Exploration, Reserve, Consumption, and Price Paths.}  \label{samplePath.fig} 
\includegraphics[width=0.49\textwidth]{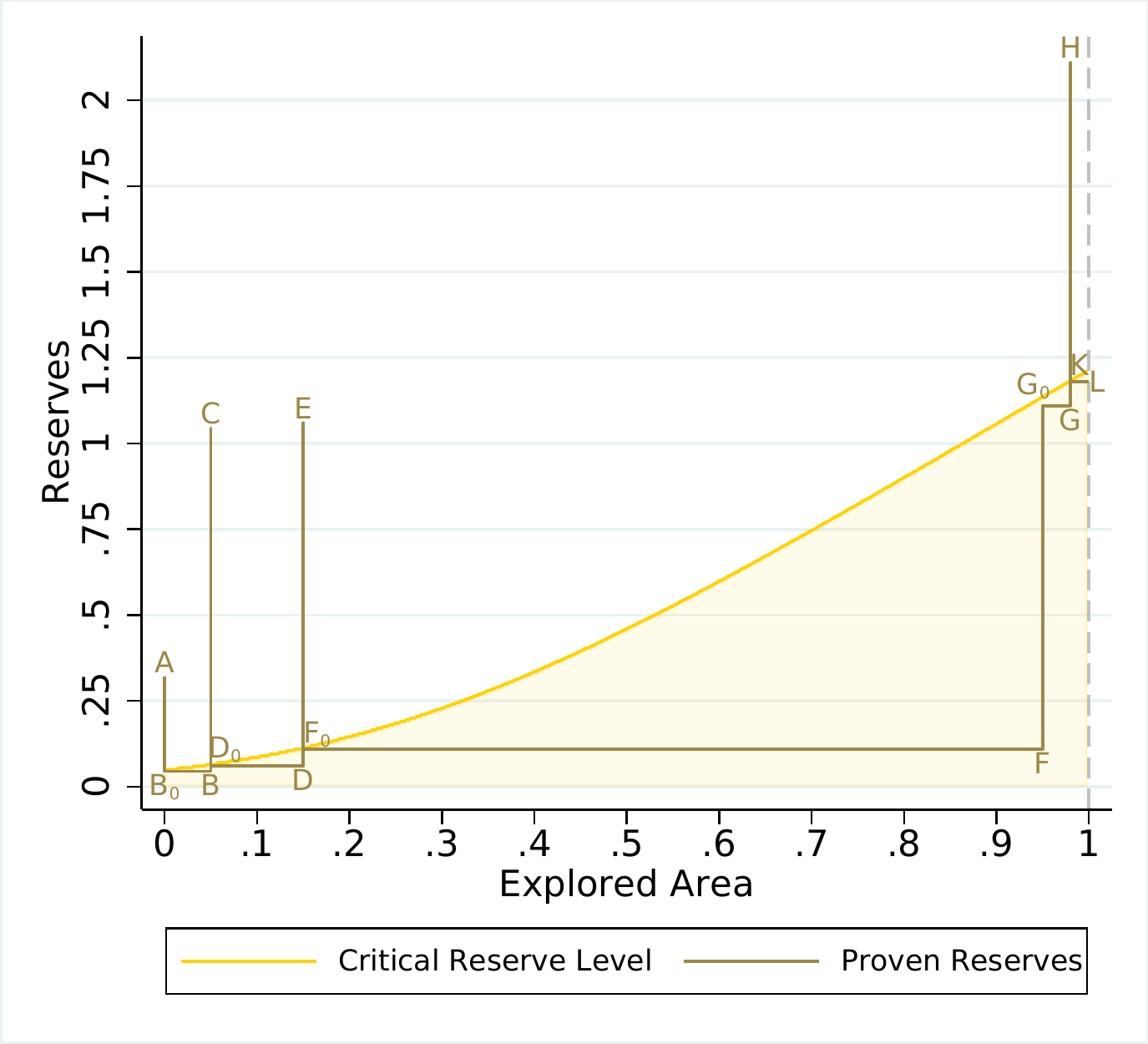}
\includegraphics[width=0.49\textwidth]{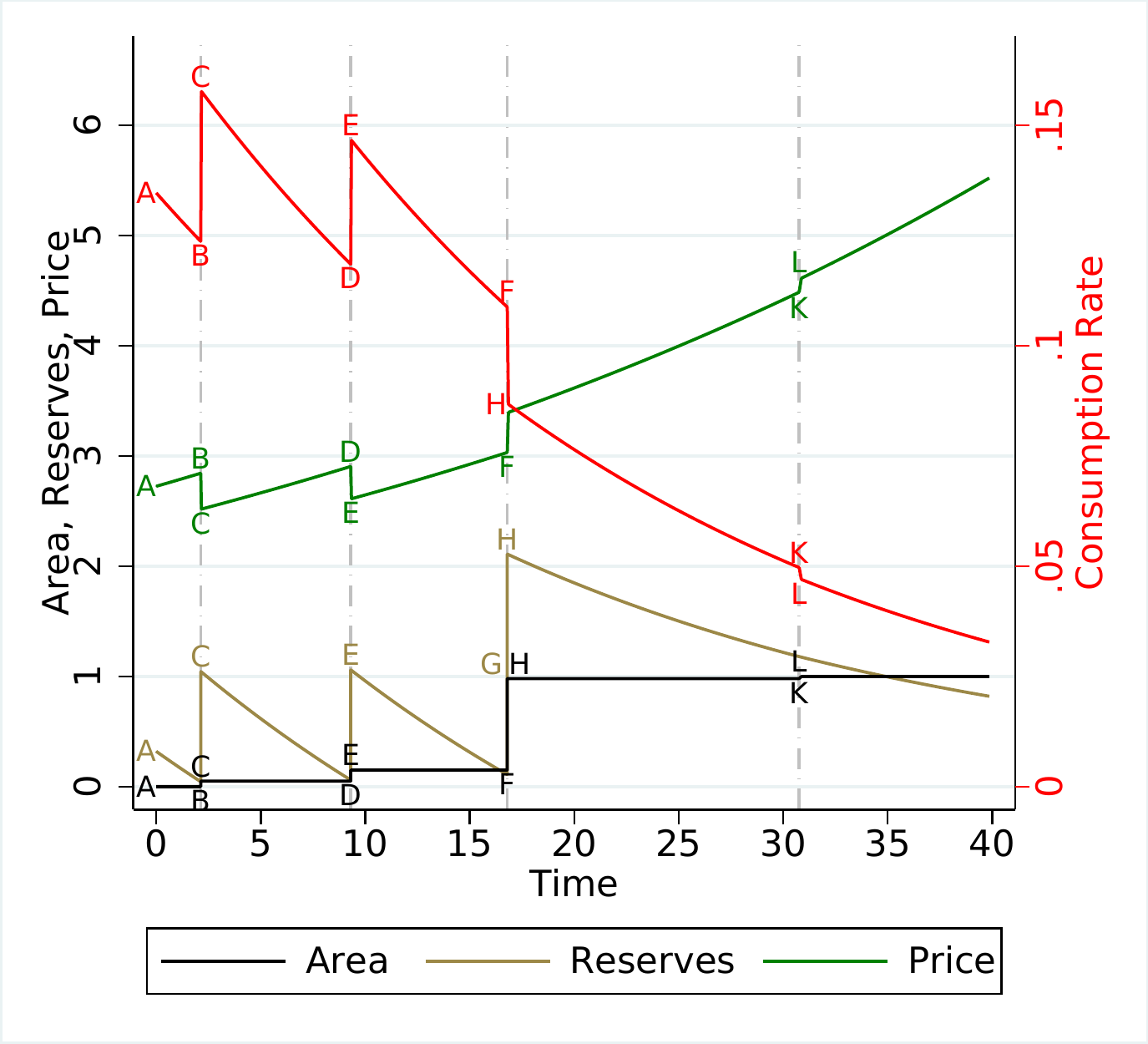}
\emph{Notes}: Left graph shows proven reserves $R$ against explored area as well as the critical reserve level when exploration starts in yellow. Right graph plots explored area, proven reserves $R$, price $P$ and consumption rate $c$ over time. Parameters are $\alpha= 0.5, r=0.02, a=1, \lambda=5, k=5$.
\end{figure}

\begin{figure} 
\caption{Critical Reserve Level and Price when Exploration Starts} \label{region1.fig}
\includegraphics[width=0.49\textwidth]{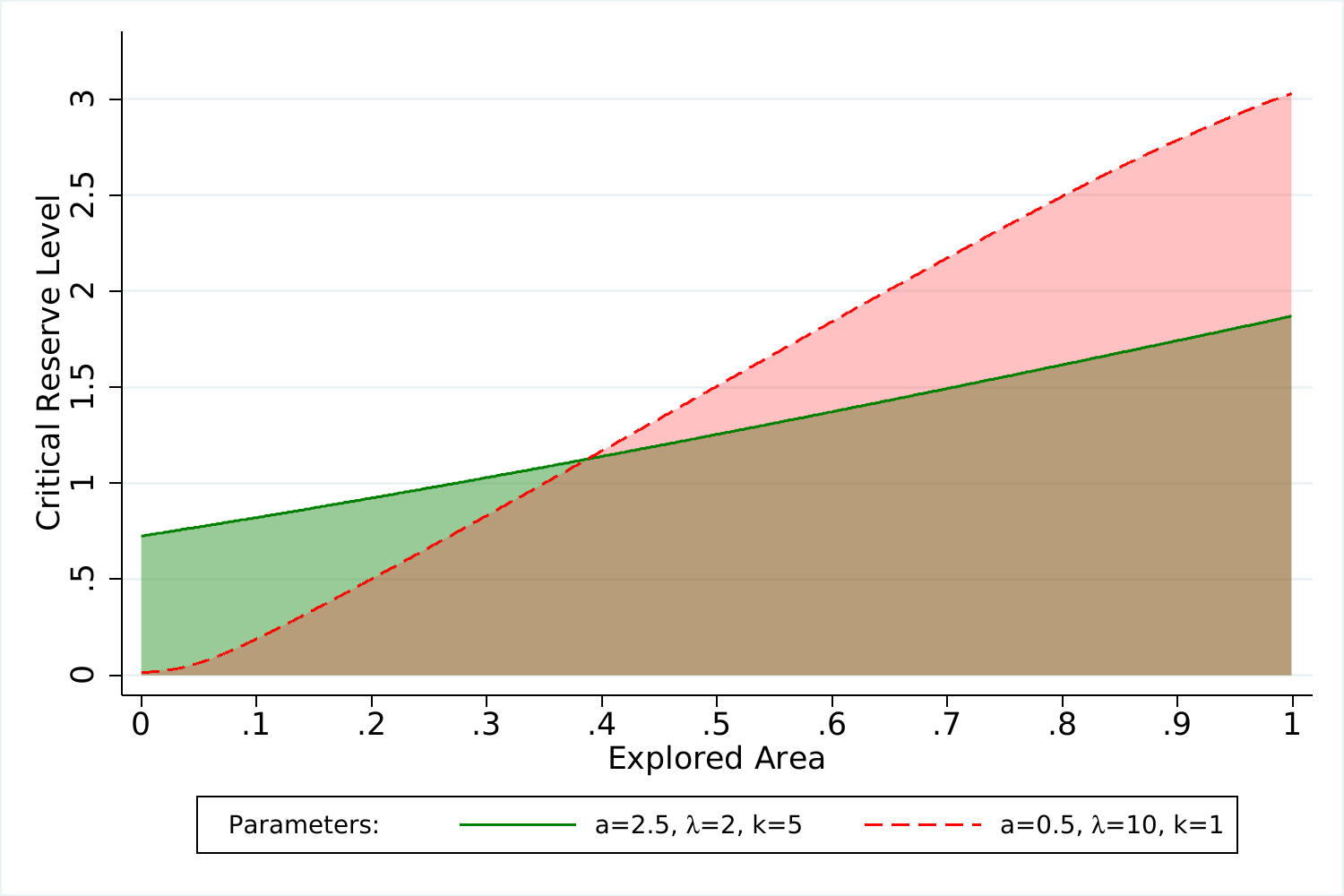}
\includegraphics[width=0.49\textwidth]{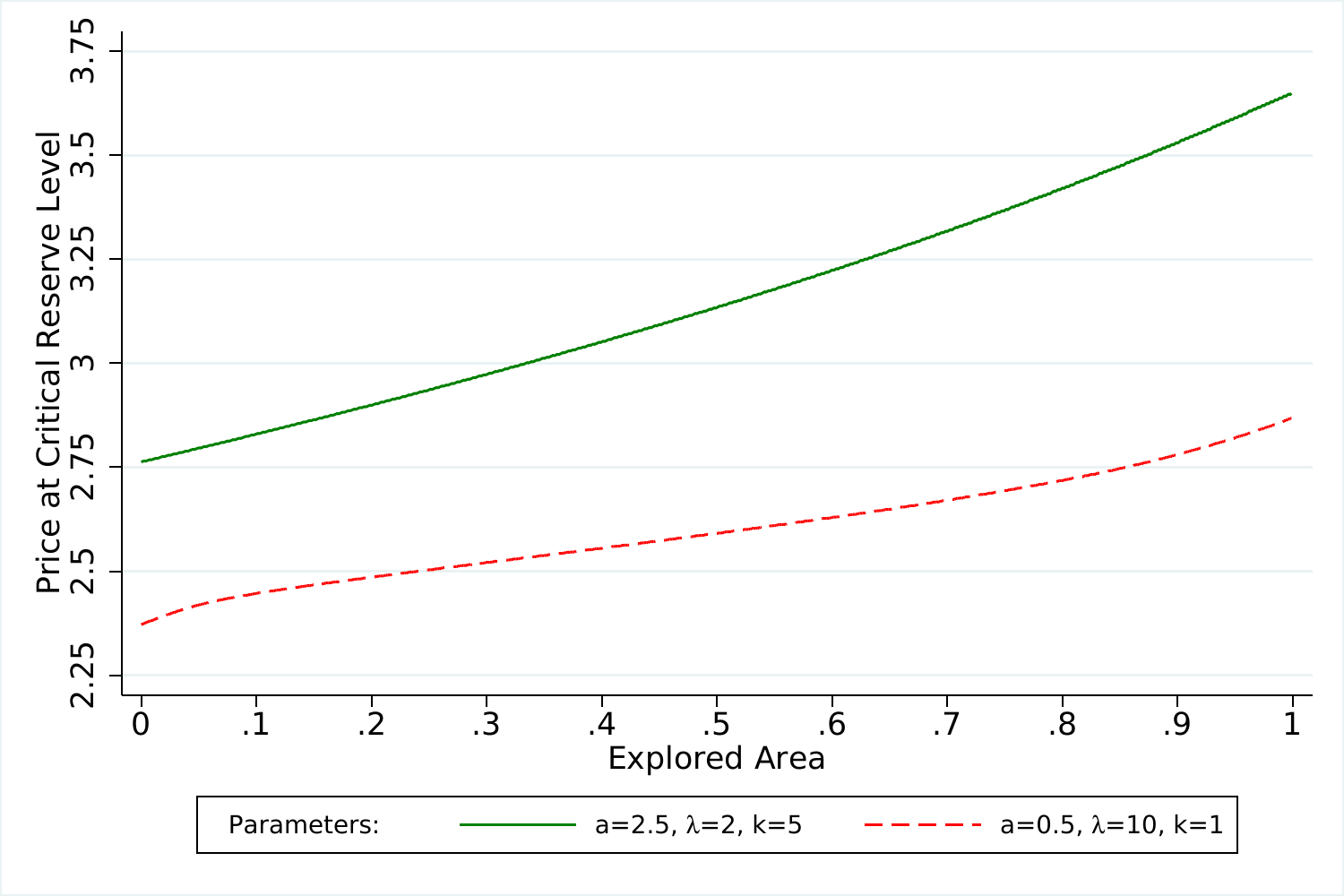}
\label{frontiers.fig}
\emph{Notes}: Left graph shows the critical reserve level when exploration starts under two different parameter sets. Exploration occurs at infinite speed (zero time) when the known reserves are below the critical reserve level, shown by the shaded area. The right graph shows the price at the critical reserve level. Note that the two sets of parameters have the same expected number of findings, but the variance is higher in case of the green line.
\end{figure}

\clearpage
\begin{figure}
\caption{Distribution of Variables - High Likelihood of Discoveries, Which are Small}  \label{price10.fig}
\includegraphics[width=\textwidth]{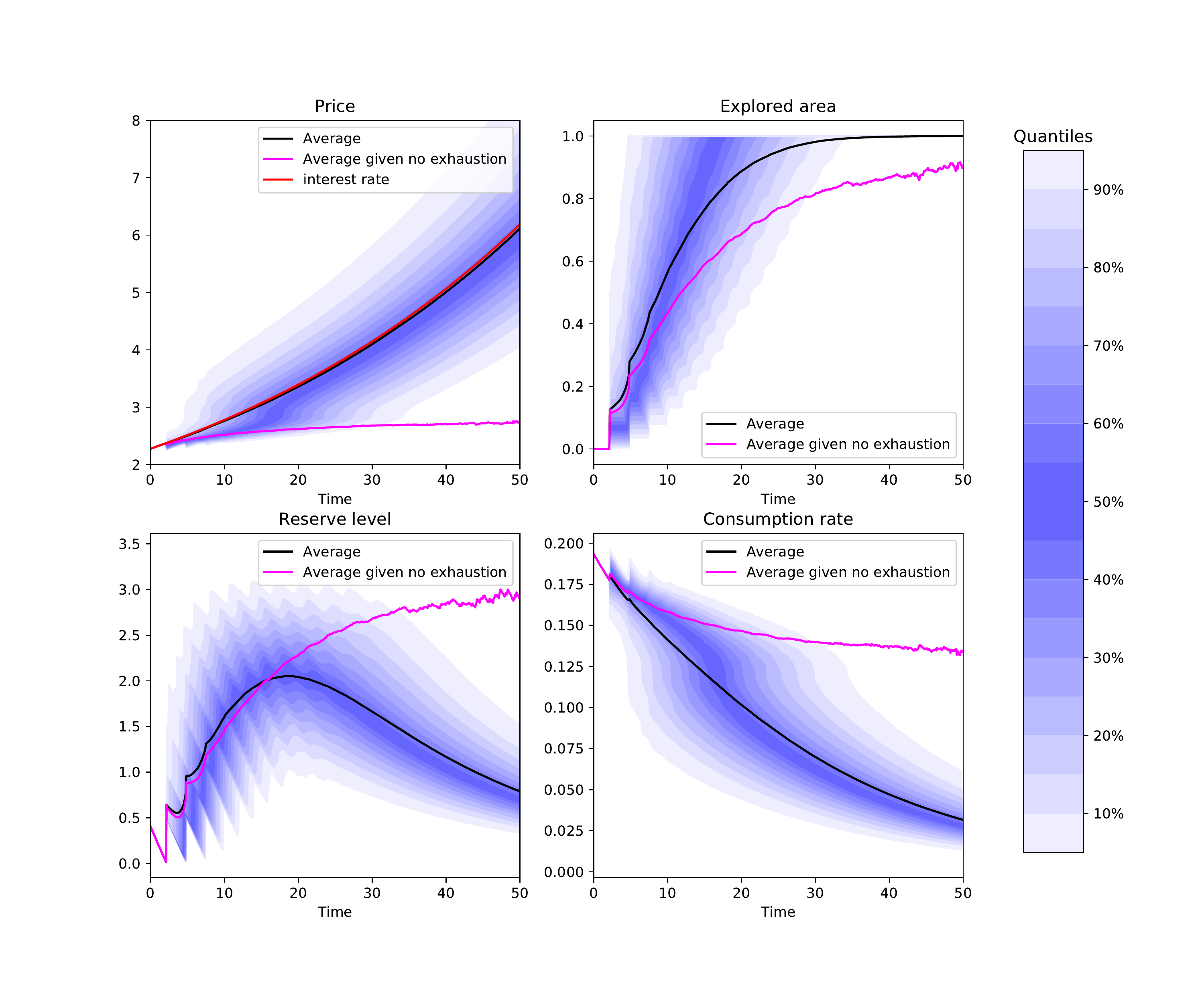}
\emph{Notes}: Graphs shows 1000 simulations of prices, explored area, reserve levels, and the consumption rate over time under the parameter assumptions with a high probability of discoveries (red line in Figure~\ref{region1.fig}).
\end{figure}

\clearpage
\renewcommand{\thepage}{\roman{page}}
\renewcommand{\thefigure}{A\arabic{figure}}
\renewcommand{\thetable}{A\arabic{table}}
\renewcommand{\theequation}{A\arabic{equation}}
\renewcommand{\thesection}{A\arabic{section}}
\setcounter{page}{1}
\setcounter{table}{0}
\setcounter{figure}{0}
\setcounter{section}{0}
\singlespace
	
\section{Appendix: Proofs} \label{proofs.sec}
\subsection{A priori results}
\begin{proof}[Proof of Proposition \ref{apriori.prop}] \textbf{Part i.}\ Let $(c,\theta,h)\in\mathcal A(x,R)$ and consider a consumption-exploration strategy $(\tilde c,\tilde\theta)$, which follows the same consumption profile as $(c,\theta)$, after having first explored all unexplored area at time $t=0$. Mathematically, this strategy writes: $\tilde\theta_n = 0$ and $\tilde c^n(t) = c_t \mathbf1_{n = N_x+1}$ for all $n\geq 0$.  Clearly, the new strategy is admissible: $(\tilde c,\tilde \theta) \in \mathcal A(x,R)$, and the consumption processes are equal: $\tilde c_t = c_t$ for all $t\geq 0$. Thus, the value function equals
$$
\sup_c \mathbb E\left[\int_0^\infty e^{-rt} u(c_t)
  dt\right],
$$
where the sup is taken over all mappings $c:\Omega\times \mathbb R_+$ which are $\mathcal F_{N_x+1}\times \mathcal B(\mathbb R_+)$-measurable and satisfy the admissibility condition
$$
R+aN_{x} - \int_0^t c_s ds\geq 0
$$
for all $t\geq 0$. Similarly to the original Hotelling problem, this supremum is found to be equal to $\mathbb E[\mathcal U(R + a N_{x})]$, from which the upper bound follows. 

\textbf{Part ii.}\ Let us first show that the value function is increasing in $R$. If a consumption-exploration strategy $(c,\theta)$ is admissible for initial data $(x,R)$, it is also admissible for initial data $(x,R')$ for all $R'\geq R$. Since the gain of a fixed consumption-exploration strategy does not depend on $R$, this shows that $V$ is increasing in $R$. Now, let $\delta>0$ and let $(c,\theta)\in \mathcal A(x,R+\delta)$ be a consumption-exploration strategy.  Then, clearly, 
$$
R - \int_0^t\frac{R}{R+h} c_s ds + a \sum_{n=0}^{N_{x}-1} \mathbf 1_{\theta_n\leq
t} \geq 0
$$
a.s. for all $t\geq 0$, so that $\left(\frac{R}{R+\delta} c,\theta\right)\in \mathcal A(x,R)$. Then,
\begin{eqnarray*}
V(x,R+\delta)-V(x,R) & \leq  & \sup_{(c,\theta)\in \mathcal A(x,R+\delta)}\mathbb E\left[\int_0^\infty
  e^{-rt}\left(u(c_t) -
                  u\left(\frac{R}{R+\delta}c_t\right)\right)dt\right]\\
& = & \left(1 - \left(\frac{R}{R+\delta}\right)^\alpha\right)\sup_{(c,\theta)\in \mathcal A(x,R+\delta)}\mathbb E\left[\int_0^\infty
  e^{-rt}u(c_t) dt\right]\\
&=& \left(1 - \left(\frac{R}{R+\delta}\right)^\alpha\right)\mathbb E[ U(R+\delta +
aN_{x})]\\
& \leq& \frac{\alpha \delta}{R} \mathbb E[ U(R+\delta +
aN_{x})].
\end{eqnarray*}
Since $V$ is increasing, it follows that it is locally Lipschitz continuous.

\textbf{Part iii.}\ Let us first show the upper bound. From equation (\ref{vf}), part i.~of the present proposition, and equation (\ref{uhotel}), it follows that
$$
V(x,R) \leq \sup_{(c,\theta)\in \mathcal A(x,R)} \mathbb E\left[\int_0^\infty e^{-rt} u(c_t)
  dt\right] = \mathbb E[\mathcal U(R + aN_{x})] = C \mathbb E[(R + aN_{x})^\alpha].
$$
for some $C<\infty$. To evaluate the expectation in the right-hand side, recall that $\alpha\in (0,1)$. Then, 
\begin{eqnarray*}
\mathbb E[(R+aN_{x})^\alpha] &\leq& R^\alpha + \mathbb E[(aN_{x})^\alpha]\\ &=& R^\alpha + \sum_{n=0}^\infty \frac{e^{-\lambda{x}}(\lambda{x})^n}{n!} (an)^{\alpha}\\
&=& R^\alpha + \sum_{n=0}^\infty \frac{e^{-\lambda{x}}(\lambda{x})^{n+1}}{n!} \frac{(a(n+1))^{\alpha}}{n+1}\\
&\leq & R^\alpha + \lambda{x}a^\alpha,
\end{eqnarray*}
which shows that the upper bound of the statement holds with a different constant $C$. 

To prove the lower bound, consider a strategy $(c,\theta)$, where $\theta_n = +\infty$ for all $n$, $c_0$ is given by Hotelling's rule and $c_n$ for $n\geq 1$ are arbitrary. This strategy consists in not doing any exploration and simply consuming the reserves available at time $t=0$ according to Hotelling's rule. Since is clearly admissible for initial data $(x,R)$ with any $x\geq 0$, this shows that the value function $V(x,R)$ is bouded from below by the Hotelling value function $\mathcal U(R)$ for all $x\geq 0$. 

\end{proof}

\subsection{A first approach at the HJB equation}
We start by deriving the following dynamic programming principle for our value function. 
\begin{lemma}[Dynamic programming principle]\label{dynpro}
The value function satisfies:
$$
V(x,R) = \sup_{c,\theta_1} \int_0^{\theta_1} e^{-rt}u(c_t) dt +
e^{-r\theta_1} M V(x,R-\int_0^{\theta_1} c_s \, ds),
$$
where the $\sup$ is taken over all measurable deterministic functions
$c:\mathbb R_+ \to \mathbb R_+$ and constants $\theta_1\in \mathbb R_+$  such that $\int_0^{\theta_1} c_s ds\leq R$.
\end{lemma}
\begin{proof}
  We denote by $\mathcal A^{n}(x,R)$ the set of consumption-exploration strategies in $\mathcal A(x,R)$ such that $\theta_k=+\infty$ for all $k\geq n$, and by $V^n(x,R)$ the value function with at most $n$ exploration dates, defined as follows:
  $$
  V^n(x,R) = \sup_{(c,\theta)\in \mathcal A^n(x,R)} \mathbb E\left[\int_0^\infty e^{-rt} u(c_t)
  dt- k\sum_{j=0}^{N_{x}} e^{-r\theta_j}(X_{j+1} - X_j)\right].
$$
In particular, $V^0(x,R) = \mathcal U(R)$ for all $R\geq 0$. Also, it is easy to see that
\begin{eqnarray*}
  V^n(x,R) &=& \sup_{(c,\theta)\in \mathcal A^n(x,R)} \mathbb E\Bigg[\int_0^{\theta_{n-1\wedge N_{x}}} e^{-rt} u(c_t)
  dt + e^{-r\theta_{n-1\wedge N_{x}}} \mathcal U(R_{\theta_{n-1\wedge N_{x}}})\\ &-& k\sum_{j=0}^{n-1\wedge N_{x}} e^{-r\theta_j}(X_{j+1} - \xi_j)\Bigg].
\end{eqnarray*}
The rest of the proof is divided into three steps.


\textit{Step 1.}\ In this step we show that the value function $V^n(x,R)$ for $n\geq 1$ satisfies the dynamic programming principle. To this end, introduce the function $\widetilde V^n$ as follows:
\begin{equation}
\widetilde V^n(x,R) = \sup_{c,\theta} \int_0^{\theta} e^{-rt}u(c_t) dt +
e^{-r\theta} M\widetilde V^{n-1}(x,R-\int_0^{\theta} c_s \, ds),\label{dpn}
\end{equation}
with $\widetilde V^0 = V^0$, where the supremum is taken over all $\theta \in [0,\infty]$ and over all measurable functions $c:[0,\infty)\mapsto [0,\infty)$ such that $\int_0^\theta c_s ds \leq R$. It is easy to see that for every $n\geq 0$, $\widetilde V^n$ is continuous in $R$ and bounded from above by $C(1+R^\alpha)$ for some constant $C$. It is also not difficult to show, using the method of Lagrange multipliers, that the supremum above is attained, namely
$$
c_t = \frac{Q^*r}{1-\alpha} \frac{e^{-\frac{rt}{1-\alpha}}}{1-e^{\frac{r\theta^*}{1-\alpha}}}\mathbf 1_{t\leq \theta^*},
$$
where $\theta^*$ and $Q^*$ are given by
$$
(\theta^*,Q^*) = \arg\max_{0\leq Q\leq R, \theta\in[0,\infty]} \{\widetilde U(\theta,Q) +
e^{-r\theta} M\widetilde V^{n-1}(x,R-Q)\},
$$
with
$$
\widetilde U(\theta_1,Q) =
\frac{Q^\alpha}{\alpha}\left(\frac{1-\alpha}{r}\right)^{1-\alpha}(1-e^{-\frac{r\theta_1}{1-\alpha}})^{1-\alpha}. 
$$

Now, let $p< n$ and $(c,\theta)\in \mathcal A^n(x,R)$. Then, on $N_x\geq p$,
$$
\mathbb E[\widetilde V^{n-p-1}(x-X_{p+1},R_{\theta_p})|\mathcal F_{X_p}] = MV^{n-p-1}(x-X_{p},R_{\theta_p-})+k\mathbb E[X_{p+1}-X_{p}|\mathcal F_{X_p}].
$$
Therefore, by the law of iterated expectations, 
\begin{eqnarray*}
  && \mathbb E\Big[\int_0^{\theta_p\wedge N_{x}}e^{-rt} u(c_t) dt - k \sum_{j=0}^{p\wedge N_{x}}e^{-r\theta_j} (X_{j+1}-X_j) \\ &&\qquad \qquad\qquad\qquad+ e^{-r\theta_{p\wedge N_{x}}} \widetilde V^{n-p-1} (x- X_{p+1\wedge N_x+1},R_{\theta_{p\wedge N_{x}}})\Big]\\
  &&= \mathbb E\left[\int_0^{\theta_{p-1}\wedge N_{x}}e^{-rt} u(c_t) dt - k \sum_{j=0}^{{p-1}\wedge N_{x}}e^{-r\theta_j} (X_{j+1}-X_j) \right] \\
  &&\qquad \qquad + \mathbb E\left[\mathbf 1_{N_{x}\geq p}\Big\{\int_{\theta_{p-1}}^{\theta_p}e^{-rt} u(c_t) dt + e^{-r\theta_p} M\widetilde V^{n-p-1}\Big(x-X_p, R_{\theta_{p-1}} - \int_{\theta_{p-1}}^{\theta_{p}} c_s ds\Big)\Big\}\right]\\
 &&\qquad \qquad+ \mathbb E\left[\mathbf 1_{N_{x}< p} e^{-r\theta_{p-1\wedge N_{x}}} \widetilde V^{n-p} (x-X_{p\wedge N_x+1},R_{\theta_{p-1\wedge N_{x}}})\right] 
 \end{eqnarray*}
where in the last term we could replace $\widetilde V^{n-p-1}$ with $\widetilde V^{n-p}$ because on $N_x<p$, $X_{p\wedge N_x+1} = X_{N_{x}+1} = x$.

Therefore, by equation (\ref{dpn}), for any strategy,
\begin{eqnarray*}
  && \mathbb E\Big[\int_0^{\theta_p\wedge N_{x}}e^{-rt} u(c_t) dt - k \sum_{j=0}^{p\wedge N_{x}}e^{-r\theta_j} (X_{j+1}-X_j) \\ &&\qquad \qquad\qquad\qquad+ e^{-r\theta_{p\wedge N_{x}}} \widetilde V^{n-p-1} (x-X_{p+1\wedge N_x+1},R_{\theta_{p\wedge N_{x}}})\Big]\\
   &&\leq \mathbb E\Big[\int_0^{\theta_{p-1}\wedge N_{x}}e^{-rt} u(c_t) dt - k \sum_{j=0}^{{p-1}\wedge N_{x}}e^{-r\theta_j} (X_{j+1}-X_j)  \\ &&\qquad \qquad \qquad \qquad+ e^{-r\theta_{p-1\wedge N_x}} \widetilde V^{n-p}(x - X_{p\wedge N_x+1}, R_{\theta_{p-1\wedge N_x}})\Big]. 
\end{eqnarray*}
On the other hand, for the optimal strategy which achieves the supremum in (\ref{dpn}), inequality becomes equality.

Iterating this expression from $p=n-1$ to $p=0$, we finally obtain that for any strategy $(c,\theta)\in \mathcal A^n(x,R)$,
\begin{eqnarray*}
\mathbb E\Bigg[\int_0^{\theta_{n-1\wedge N_{x}}} e^{-rt} u(c_t)
  dt + e^{-r\theta_{n-1\wedge N_{x}}} \mathcal U(R_{\theta_{n-1\wedge N_{x}}})\\ - k\sum_{j=0}^{n-1\wedge N_{x}} e^{-r\theta_j}(X_{j+1} - X_j)\Bigg]\leq \widetilde V^n(x,R),
\end{eqnarray*}
with equality for the optimal strategy. Therefore, $V^n(x,R) =  \widetilde V^n(x,R)$ for all $x,R$.

\textit{Step 2.}\ In this step our goal is to show that for all $x\geq 0$ and $R\geq 0$, $V^n(x,R)\to V(x,R)$ as $n\to\infty$. For an admissible bang-bang consumption-exploration strategy $(c,\theta)\in \mathcal A(x,R)$ and initial data $(x,R)$, we define
$$
J(c,\theta) := \mathbb E\left[\int_0^\infty u(c_t) dt - k\sum_{n=0}^{N_x} e^{-r\theta_n} (X_{n+1}-X_n)\right].
$$
By definition of the value function, for every $\varepsilon>0$, there exists an admissible strategy $(c,\theta)\in \mathcal A(x,R)$ such that 
$$
J(c,\theta)\geq V(x,R) - \varepsilon. 
$$
Define a strategy $(c^n,\theta^n)\in \mathcal A^n(x,R)$ by taking $c^n_j = c_j$ and $\theta^n_j = \theta_j$ for $j=0,\dots,n$ and $c^n_j = 0$ and $\theta^n_j = +\infty$ for $j>n$. Then, since $N_x$ is finite, 
$$
\int_0^\infty u(c^n_t) dt - k\sum_{j=0}^{n-1} e^{-r\theta^n_j} (X^n_{j+1}-X_j)
$$
converges to 
$$
\int_0^\infty u(c_t) dt - k\sum_{n=0}^{N_x} e^{-r\theta_n} (X_{n+1}-X_n)
$$
as $n\to \infty$, the two expressions are actually equal for $n$ large enough. Then, by the dominated convergence theorem,
$$
J(c^n,\theta^n) \to J(c,\theta)
$$
as $n\to \infty$. This shows that $V^n(x,R)\geq V(x,R)-\varepsilon$, and since on the other hand $V^n(x,R)\leq V(x,R)$ and $\varepsilon$ is arbitrary, we conclude that $V^n(x,R)\to V(x,R)$ as $n\to \infty$. Moreover, since the sequence $(V^n(x,R))_{n\geq 1}$ is increasing, we conclude using Dini's theorem that the convergence is uniform in $n$.

\textit{Step 3.}\ It remains to prove the DPP for the original value function by passing to the limit. From Step 1,
$$
V^n(x,R) = \sup_{c,\theta} \int_0^{\theta} e^{-rt}u(c_t) dt +
e^{-r\theta} M V^{n-1}(x,R-\int_0^{\theta} c_s \, ds).
$$
From Step 2, as $n\to\infty$, $V^n(x,R)\to V(x,R)$, and the convergence is uniform in $x$ and uniform on compacts in $R$. This implies that $MV^{n-1}$ converges to $MV$, also uniformly in $x$ and uniformly on compacts in $R$, and so the supremum converges as well. 
\end{proof}

\begin{corollary}\label{cor1}
For all $\delta>0$, $\hat c:[0,\delta]\to \mathbb R_+$, such that $\int_0^\delta \hat
c_s ds\leq R$, 
$$
V(x,R)\geq \int_0^\delta e^{-rt} u(\hat c_t) dt + e^{-r\delta} V\left(x,R-\int_0^\delta
  \hat c_s ds\right).
$$
\end{corollary}
\begin{proof}
Let $\tilde c:\mathbb R_+ \to \mathbb R_+$ and $\tilde\theta_1\in \mathbb
R_+$. Applying the dynamic programming principle  to 
\begin{eqnarray*}
&c:&  \mathbb R_+ \to \mathbb R_+,\quad t\mapsto \hat c_t \mathbf 1_{t<\delta}
+ \tilde c_{t+\delta} \mathbf 1_{t\geq \delta}\\
&\theta_1 & = \delta + \tilde \theta_1,
\end{eqnarray*}
we get:
{\footnotesize\begin{eqnarray*}
V(x,R) & \geq & \int_0^\delta e^{-rt} u(\hat c_t) dt +
\int_{\delta}^{\delta+\tilde\theta_1} e^{-rt} u(\tilde c_{t+\delta}) dt +
e^{-r(\delta+\tilde\theta_1)} MV\left(x,R - \int_0^\delta \hat c_t dt +
  \int_\delta^{\delta+\tilde\theta_1} \tilde c_{t+\delta} dt\right)\\
& = & \int_0^\delta e^{-rt} u(\hat c_t) dt +
e^{-r\delta}\left\{\int_{0}^{\tilde\theta_1} e^{-rt} u(\tilde c_{t}) dt +
e^{-r\tilde\theta_1} MV\left(x,R - \int_0^\delta \hat c_t dt +
  \int_0^{\tilde\theta_1} \tilde c_{t} dt\right)\right\}.
\end{eqnarray*}}
Taking the sup over $\tilde c$ and $\tilde \theta_1$, we get the
statement of the corollary. 
\end{proof}

The following corollary follows easily by the method of Lagrange multipliers. 
\begin{corollary}\label{cor2} The value function satisfies:
$$
V(x,R) = \sup_{0\leq Q\leq R,\theta_1\geq 0} \{\widetilde U(\theta_1,Q) +
e^{-r\theta_1} MV(x,R-Q)\},
$$
where
$$
\widetilde U(\theta_1,Q) =
\frac{Q^\alpha}{\alpha}\left(\frac{1-\alpha}{r}\right)^{1-\alpha}(1-e^{-\frac{r\theta_1}{1-\alpha}})^{1-\alpha}. 
$$
\end{corollary}

The following proposition is our first variant of the HJB equation, which does not yet require the value function to be everywhere differentiable. The differentiability of the value function will be shown in the following subsection.
\begin{proposition}\label{HJB.prop}
The value function $V(x,R)$ is the solution of the
  HJB equation
$$
\max\left\{u^*\left(\frac{\partial V}{\partial R}\right) -
    rV,MV-V\right\}=0,\quad V(1,R) = U(R),
$$
meaning that
\begin{itemize}
\item[i.] For all $R\geq 0$, $x\in [0,1]$, $V(x,R)\geq MV(x,R)$. 
\item[ii.] For all $R>0$, $x\in [0,1]$, 
$$
u^*\left(\liminf_{h\to 0} \frac{V(x,R+h)-V(x,R)}{h}\right)\leq
rV(x,R). 
$$ 
\item[iii.] At all points $(x,R)$ such that $R>0$ and
  $V(x,R)>MV(x,R)$, $V$ is differentiable in $R$ and satisfies
$$
u^*\left(\frac{\partial V}{\partial R}\right) = rV. 
$$
\end{itemize}
\end{proposition}
\begin{proof}${}$
It follows from Lemma \ref{dynpro} that $V\geq MV$. Let us prove the
property ii. In other words, we need to show that for all $c>0$, 
\begin{equation}
u(c) - c\liminf_{h\to 0} \frac{V(x,R+h)-V(x,R)}{h}  \leq rV.\label{liminf}
\end{equation}
Fix $\delta<R$ and $\bar c>0$. From Corollary \ref{cor1}, 
$$
V(x,R) \geq \int_0^{\delta/\bar c} e^{-rt} u(\bar c) dt  + e^{-r\delta/\bar c}
V(x,R-\delta),
$$
and therefore 
$$
\frac{V(x,R+\delta)-V(x,R)}{\delta}\geq u(\bar c)\frac{1-e^{-r\delta/\bar c}}{\delta r} -
\frac{1-e^{-r\delta/\bar c}}{\delta}V(x,R),
$$
from which equation~(\ref{liminf}) follows by passing to the limit. 

Let us now turn to property iii. By Corollary \ref{cor2} and
continuity of $V$ and $MV$, there exists $\varepsilon>0$ such that
$V(x,R')> M V(x,R')$ for all $R'$ with $|R-R'|<\varepsilon$ and 
$$
V(x,R) = \sup_{\varepsilon\leq Q\leq R, \theta_1\geq 0} \{\widetilde U(\theta_1,Q) +
e^{-r\theta_1} MV(x,R-Q)\}.
$$
Then,
\begin{eqnarray*}
V(x,R) & = & \sup_{c,\theta_1\geq 0: \int_0^{\theta_1} c_s ds >\varepsilon} \int_0^{\theta_1} e^{-rt}u(c_t) dt +
e^{-r\theta_1} MV(x,R-\int_0^{\theta_1} c_s \, ds)\\
& = & \sup_{c,\theta_1\geq \tau: \int_0^{\tau} c_s ds =\varepsilon} \int_0^{\tau} e^{-rt}u(c_t) dt +\int_\tau^{\theta_1} e^{-rt}u(c_t) dt+
e^{-r\theta_1} MV(x,R-\varepsilon-\int_\tau^{\theta_1} c_s \, ds)\\
& = & \sup_{c,\tau \tau: \int_0^{\tau} c_s ds =\varepsilon}
  \int_0^{\tau} e^{-rt}u(c_t) dt +e^{-r\tau} V(x,R-\varepsilon) \\
 & = &
  \sup_{\tau\geq 0}\{\widetilde U(\tau,\varepsilon) + e^{-r\tau}
  V(x,R-\varepsilon)). 
\end{eqnarray*}
Remark that 
$$
\widetilde U(\tau,Q) \leq \frac{Q^\alpha}{\alpha r^{1-\alpha}} (1-e^{-r\tau})^{1-\alpha}.
$$
Thus,
$$
V(x,R) \leq \sup_\tau \{\frac{\varepsilon^\alpha}{\alpha r^{1-\alpha}}
(1-e^{-r\tau})^{1-\alpha} +
e^{-r\tau} V(x,R-\varepsilon)\}
$$
The first-order condition for the maximization in the RHS writes:
$$
1-e^{-r\tau} = \varepsilon\left(\frac{V(x,R-\varepsilon) \alpha r^{1-\alpha}}{1-\alpha}\right)^{-\frac{1}{\alpha}},
$$
which provides an upper bound:
$$
\frac{V(x,R) - V(x,R-\varepsilon)}{\varepsilon} \leq \left(\frac{\alpha}{1-\alpha} rV(x,R-\varepsilon)\right)^{\frac{\alpha-1}{\alpha}}.
$$
On the other hand, from property ii., we have 
$$
\lim\inf_{\varepsilon\to 0} \frac{V(x,R)-V(x,R-\varepsilon)}{h} \geq
\left(\frac{\alpha}{1-\alpha} rV(x,R-\varepsilon)\right)^{\frac{\alpha-1}{\alpha}}.
$$
Together, the two inequalities show the differentiability of $V$ and
complete the proof of the first part. 

\end{proof}


\subsection{Characterization of consumption / exploration regions, smoothness of the value function}
Recall that $u^*$ denotes the the convex conjugate of $u$. We shall denote by $u_1$ the
inverse of $u^*$. This function is given explicitly by
$$
u_1(y) = \left(\frac{\alpha y}{1-\alpha}\right)^{1- \frac{1}{\alpha}}
$$
The following result establishes some useful properties of the exploration operator $M$ defined in (\ref{exp.op}). 
\begin{lemma}\label{cinf}
For all $(x,R) \in \mathbb R^2_+$, $MV(x,R)$ is infinitely differentiable in $R$ and satisfies
\begin{eqnarray*}
MV(x,R) & = & \mathbb E[V(x-X_{\tau_C},\widehat R(\tau_C))- k X_{\tau_C}]\\
MV^{(n)}_R(x,R) & = & \mathbb E[V^{(n)}_R(x-X_{\tau_C},\widehat R(\tau_C))],\quad n\geq 1,
\end{eqnarray*}
where $\tau_C = \inf\{j\geq 1: V(x-X_j,\widehat R(j))>MV(x-X_j,\widehat R(j)) \ \text{or}\ X_j=x\}$ and $\widehat R(j) = R + a(j\wedge N_{x}) $.
\end{lemma}
\begin{proof}
  By definition of the operator $M$,
  \begin{eqnarray*}
  MV(x,R) &=& \mathbb E[V(x-X_1,\widehat R(1)) - kX_1]\\
 &=& \mathbb E[(MV(x-X_1,\widehat R(1)) -kX_1)\mathbf 1_{\tau_C \geq 2} \\&&\qquad \qquad + (V(x-X_{\tau_C},\widehat R(\tau_C)) -kX_{\tau_C})\mathbf 1_{\tau_C =1}].
  \end{eqnarray*}
Iterating this procedure $n$ times, we finally obtain for $n\geq 1$:
$$
MV(x,R) = \mathbb E[(V(x-X_n,\widehat R(n)) -kX_n)\mathbf 1_{\tau_C > n} + (V(x-X_{\tau_C},\widehat R(\tau_C)) -kX_{\tau_C})\mathbf 1_{\tau_C \leq n}]
$$
Passing to the limit $n\to \infty$ using the dominated convergence theorem, we obtain the first statement of the lemma.

To prove the second statement remark that for given $(x,R)$ either $V(x-X_{\tau_C},\widehat R(\tau_C))>MV(x-X_{\tau_C},\widehat R(\tau_C)) $ or $V(x-X_{\tau_C},\widehat R_{\tau_C}) = \mathcal U(\widehat R(\tau_C))$. Therefore, by Proposition \ref{HJB.prop}, the mapping $R\mapsto V(x-X_{\tau_C},\widehat R(\tau_C))$ is differentiable in the neighborhood of $R$, with the derivative given by
$$
\frac{\partial }{\partial R}V(x-X_{\tau_C},\widehat R(\tau_C)) = u_1(r V(x-X_{\tau_C},\widehat R(\tau_C))) . 
$$
Since the latter expression is positive and bounded from above by $u_1(r \mathcal U(R))$ (see Proposition \ref{apriori.prop}, part iii.), we conclude by the dominated convergence theorem, that the second statement of the lemma is true with $n=1$. Similar arguments can be used to finish the proof for arbitrary $n$. 
\end{proof}

The following proposition uses the lemma to establish the existence of an \emph{exploration frontier}. 
\begin{proposition}\label{boundary.prop}For every $x>0$ there exists $R^*(x)\in [0,\infty]$ such that $V(x,R)>MV(x,R)$ for all $R> R^*(x)$ and $V(x,R)= MV(x,R)$ for all $R\leq R^*(x)$. 
\end{proposition}
\begin{proof}
Fix $x>0$ and let $\mathcal C_x = \{R:V(x,R)>MV(x,R)\}$. Since $V$ and $MV$ are continuous in $R$, the set $\mathcal C_x$ is open, and is therefore a union of disjoint open intervals. To prove the proposition it is enough to show that none of these intervals is bounded. By way of contradiction, assume that $(a,b)\subset \mathcal C_x$ is a bounded interval such that $V(x,a) = MV(x,a)$ and $V(x,b) = MV(x,b$). Consider the function $f:[a,b]\to \mathbb R_+$ defined by $f(R) = V(x,R) - MV(x,R)$. By Proposition \ref{HJB.prop} and Lemma \ref{cinf}, $f$ is infinitely differentiable on $(a,b)$, hence there exists $\bar R\in (a,b)$ such that $f'(\bar R) = 0$ and $f^{\prime\prime}(\bar R)\leq 0$. By Lemma \ref{cinf}, this means that
\begin{eqnarray*}
V'_R(x,\bar R) & = & \mathbb E[V'_R(x-X_{\tau_C},\widehat R(\tau_C))]\\
V^{\prime\prime}_R(x,\bar R) &\leq & \mathbb E[V^{\prime\prime}_R(x-X_{\tau_C},\widehat R(\tau_C))],
\end{eqnarray*}
where $\tau_C$ and $R(\tau_C)$ are defined starting from $\bar R$. 
On the other hand, by Proposition \ref{HJB.prop}, this is equivalent to 
\begin{eqnarray*}
V(x,\bar R)^{1-\frac{1}{\alpha}} &= & \mathbb E[V(x-X_{\tau_C},\widehat R(\tau_C))^{1-\frac{1}{\alpha}}]\\
V(x,\bar R)^{1-\frac{2}{\alpha}} &\geq & \mathbb E[V(x-X_{\tau_C},\widehat R(\tau_C))^{1-\frac{2}{\alpha}}].
\end{eqnarray*}
Let $Z = V(x-X_{\tau_C},\widehat R(\tau_C))^{1-\frac{1}{\alpha}}$. The above estimates imply that 
$$
\mathbb E[Z^{\frac{2-\alpha}{1-\alpha}}]\leq \mathbb E[Z]^{\frac{2-\alpha}{1-\alpha}},
$$
and since $Z$ is positive and the function $x\mapsto x^{\frac{2-\alpha}{1-\alpha}}$ is convex on $\mathbb R_+$, Jensen's inequality implies that $Z$ is deterministic, which is a contradiction unless $x=1$. 
\end{proof}

The following lemma shows that the exploration frontier is bounded from above. 
\begin{lemma}\label{uppercons.prop}
There exists a constant $\check R <\infty$ such that the set $\{(x,R):
x\geq 0, R\geq \check R\}$ belongs to the consumption region. 
\end{lemma}
\begin{proof}
It it enough to find a constant $\check R$ such that for all $R>\check R$,
$MV(x,R)<\mathcal U(R)$. By definition of $MV$, the
following estimate holds true. 
$$
MV(x,R) = \mathcal U(R) + \int_0^{x} \left\{V(x-h,R+a) - \mathcal U(R) - \frac{k}{\lambda}\right\}\lambda e^{-\lambda h} dh.
$$
From Proposition \ref{apriori.prop} and concavity of $ \mathcal U$, it follows then that,
\begin{eqnarray*}
MV(x,R) &\leq&  \mathcal U(R) + \int_0^{x} \mathbb E\left\{\mathcal U(R+a+a N_{x-h}) - \mathcal U(R) - \frac{k}{\lambda}\right\}\lambda e^{-\lambda h} dh\\
&\leq & \mathcal U(R) + \int_0^{x} \left\{a\mathcal U'(R)(1+\lambda(x-h)) - \frac{k}{\lambda}\right\}\lambda e^{-\lambda h} dh\
\end{eqnarray*}
Since $ \mathcal U'(x) = Cx^{\alpha-1}$ for some constant $C$, we have shown that the integral becomes negative starting from some $\check R$ 
and the proof is complete. 
\end{proof}

We now present a useful alternative characterization of the exploration frontier.
\begin{lemma}\label{alter.lm}
For all $x> 0$, the exploration frontier $R^*(x)$ is given by
 $$
R^*(x) = \inf\{R: g(x,R)< c\},
$$
where
$$
g(x,R) = \frac{d}{dR}(MV(x,R)^{\frac{1}{\alpha}})\quad \text{and}\quad c = \frac{1}{\alpha}\left(\frac{\alpha r}{1-\alpha}\right)^{1-\frac{1}{\alpha}}.
$$
\end{lemma}
\begin{proof}
  Fix $x>0$. By Proposition \ref{HJB.prop}, $MV'_R(x,R)\geq u_1(r\,MV(x,R))$ for $R\leq R^*(x)$. On the other hand, for $R>R^*(x)$, the argument used in the proof of Proposition \ref{boundary.prop} shows that $V'_R(x,R)>MV'_R(x,R)$. Since, in this region, $V'_R(x,R) = u_1(rV(x,R))< u_1(rMV(x,R))$, we have that $MV'_R(x,R)< u_1(rMV(x,R))$ and thus the consumption region is characterized as follows.
$$
R^*(x) = \inf\{R: MV'_R(x,R)< u_1(rMV(x,R))\}.
$$
Substituting the explicit expression of the function $u_1$, the proof of the lemma is complete. 
\end{proof}  

The following lemma focuses on the smoothness of the exploration frontier and the value function, and establishes an alternative HJB equation satisfied by $V$. 
\begin{lemma}\label{prop.lm}
The value function $V$ is concave and continuously differentiable in $R$ and in $x$ in its entire domain. Its derivative in $x$ satisfies
\begin{equation}
\frac{\partial V}{\partial x} = -k +  \lambda(V(x,R+a)-V(x,R))\label{dervx1}
\end{equation}
in the exploration region and
\begin{equation}
\frac{\partial V(x,R)}{\partial x}  = V(x,R)^{1-\frac{1}{\alpha}} V(x,R^*(x))^{\frac{1}{\alpha}-1} \frac{\partial MV(x,R^*(x))}{\partial x},\label{dervx2}
\end{equation}
in the consumption region. 

The value function satisfies an alternative HJB equation (\ref{hjb.classic}).

The critical exploration frontier is differentiable in $x$. 
\end{lemma}
\begin{proof}${}$\\
\textit{Step 1: smoothness and concavity of $V$ in $R$.}\  This statement is obviously true for $x=0$; fix $x>0$.  By Proposition \ref{HJB.prop} and Lemma \ref{cinf}, the value function is concave and continuously differentiable on $(0,R^*(x))$ and on $(R^*(x),\infty)$. It remains then to check that the right and left derivatives at $R^*(x)$ coincide. We may assume without loss of generality that $R^*(x)\in(0,\infty)$. Since, on $[0,R^*(x)]$, $V(x,R) = MV(x,R)$, the left-hand derivative satisfies
$$
V'_-(x,R) = MV'(x,R),
$$
while for the right-hand derivative we have, 
$$
V'_+(x,R)  = u_1(rV(x,R)).
$$
By Proposition \ref{HJB.prop}, part ii., 
$$
MV'(x,R) \geq u_1(rV(x,R)),
$$
and since $V'_+(x,R)\geq MV'(x,R)$, we also have
$$
u_1(rV(x,R))\geq MV'(x,R), 
$$
so that 
$$
MV'(x,R) = u_1(rV(x,R)) = V'_+(x,R).
$$

\medskip

\noindent\textit{Step 2: smoothness of the exploration frontier.}\ We use the characterization of Lemma \ref{alter.lm}. It is enough to show that $g$ is continuously differentiable in $x$ and $R$, and that its derivative with respect to $R$ is strictly positive. We now proceed to compute the derivatives of $g$. A direct computation using the formula for $MV$ shows:
\begin{eqnarray*}
\frac{\partial MV(x,R)}{\partial x} &=&  -k +\lambda(MV(x,R+a) - V(x,R))\\
\frac{\partial MV'_R(x,R)}{\partial x} & = &  \lambda(MV'_R(x,R+a) - V'_R(x,R)),
\end{eqnarray*}

and in particular,
\begin{eqnarray*}
\frac{\partial MV(x,R)}{\partial x}\Big|_{R = R^*(x)} &= & -k +\lambda(V(x,R^*(x)+a) - V(x,R^*(x)))\\
\frac{\partial MV'_R(x,R)}{\partial x}\Big|_{R = R^*(x)} &= &  \lambda(V'_R(x,R^*(x)+a) - V'_R(x,R^*(x))).
\end{eqnarray*}

The derivative of $g$ with respect to $R$ satisfies
$$
\frac{\partial g(x,R)}{\partial R} = \frac{1}{\alpha} MV^{\frac{1-2\alpha}{\alpha}}(x,R)\left\{\frac{1-\alpha}{\alpha} (MV'_R(x,R))^2 + MV(x,R)MV^{\prime\prime}_R(x,R)\right\},
$$
which is clearly continuous. At the point $(x,R^*(x))$, $MV'_R(x,R) = V'_R(x,R)$, and by the same argument as in Proposition \ref{boundary.prop}, $MV^{\prime\prime}_R(x,R)<V^{\prime\prime}_R(x,R)$. 
Together with the smooth pasting, this leads to the following estimate (for $x>0$). 
$$
\frac{\partial g(x,R^*(x))}{\partial R}< \frac{1}{\alpha} V^{\frac{1-2\alpha}{\alpha}}(x,R)\left\{\frac{1-\alpha}{\alpha} (V'_R(x,R))^2 + V(x,R)V^{\prime\prime}_R(x,R)\right\} = 0,
$$
since in the consumption region, 
$$
V'_R(x,R) = \left(\frac{\alpha r V(x,R)}{1-\alpha}\right)^{\frac{\alpha-1}{\alpha}}. 
$$
On the other hand, the derivative with respect to $x$ writes: 
\begin{eqnarray*}
\frac{\partial g(x,R^*(x))}{\partial x} & =& \frac{1-\alpha}{\alpha^2} MV^{\frac{1-2\alpha}{\alpha}}(x,R^*(x)) \frac{\partial MV(x,R^*(x))}{\partial x} \frac{\partial MV(x,R^*(x))}{\partial R} \\ 
& & \qquad \qquad  + \frac{1}{\alpha}MV^{\frac{1-\alpha}{\alpha}}(x,R^*(x))\frac{\partial MV'_R(x,R^*(x))}{\partial x},
\end{eqnarray*}
which is also a continuous function. 

\medskip

\noindent\textit{Step 3: smoothness and derivatives of $V$ in $x$.}\ In the interior of the exploration region, $V(x,R) = MV(x,R)$, which means that $V$ is continuously differentiable in $x$ and satisfies (\ref{dervx1}). 
In the interior of the consumption region, 
    $$
    V(x,R) = (MV(x,R^*(x))^{\frac{1}{\alpha}} + c(R-R^*(x)))^\alpha.
    $$
Since $MV$ and $R^*$ are differentiable in $x$, $V$ is differentiable in $x$ in the interior of the consumption region, and its derivative is given by
\begin{eqnarray*}
\frac{\partial V(x,R)}{\partial x} &=& V(x,R)^{1-\frac{1}{\alpha}} MV(x,R^*(x))^{\frac{1}{\alpha}-1} \frac{\partial MV(x,R^*(x))}{\partial x} \\ &+& \alpha V(x,R)^{1-\frac{1}{\alpha}}\left\{\frac{1}{\alpha} MV(x,R^*(x))^{\frac{1}{\alpha}-1}\frac{\partial MV(x,R^*(x))}{\partial R} -c\right\}\frac{\partial R(x)}{\partial x}.
\end{eqnarray*}
Since $V$ and $\frac{\partial V}{\partial R}$ are continuous, this shows that the second term above is zero, and so in the consumption region, the derivative of $V$ satisfies (\ref{dervx2}) and is continuous across the frontier. 

\medskip

\noindent\textit{Step 4: alternative HJB equation.}\ To finish the proof, it remains to show that in the consumption region, 
$$
\frac{\partial V}{\partial x}\geq\lambda(V(x,R+a)-V(x,R)) - k, 
$$
or, in other words, given the result of Step 3, that 
$$
 \left(\frac{V(x,R^*(x))}{V(x,R)}\right)^{\frac{1}{\alpha}-1} (\lambda(V(x,R^*(x)+a)-V(x,R^*(x))) - k) \geq \lambda(V(x,R+a)-V(x,R)) - k.
$$
Since $V$ is increasing in $R$, it is sufficient to prove that 
$$
 \left(\frac{V(x,R^*(x))}{V(x,R)}\right)^{\frac{1}{\alpha}-1} (V(x,R^*(x)+a)-V(x,R^*(x))) \geq V(x,R+a)-V(x,R),
$$
or in other words that the mapping 
$$
R\mapsto V(x,R)^{\frac{1}{\alpha}-1}(V(x,R+a)-V(x,R))
$$
is decreasing for $R\geq R^*(x)$. Equivalently, using the explicit form of $V$ in the consumption region, one can consider the mapping
$$
V \mapsto V^{\frac{1}{\alpha}-1}((V^\frac{1}{\alpha} + ac)^\alpha-V),
$$
whose derivative is easily shown to be negative. 
\end{proof}

We now provide the proof of the second main theorem of the paper.
\begin{proof}[Proof of Theorem \ref{cons.thm}]
  The existence of the exploration frontier $R^*(x)$ was shown in Proposition \ref{boundary.prop}, and in Lemma \ref{prop.lm} we have shown that the frontier is smooth. It remains to prove that $R^*(x)$ is  strictly positive and decreasing  in $x$, and that the value function is increasing in $x$, as well as to study the behavior of the frontier for small $x$. 

\textit{Step 1. $R^*(x)$ is strictly positive.}\ We will show that there exists a function $\overline R:(0,\infty)\mapsto (0,\infty)$, such that for every $x>0$,
the points $\{(x,R): R\leq \overline R(x)\}$ belong to the exploration
region. Let us define
$$
\overline R(x) = \max\{R>0: (1-e^{-\lambda x}) ( \mathcal U(R+a) -
k/\lambda) - \{e^{\frac{\alpha\lambda}{1-\alpha}x} -
e^{-\lambda{x}}\} \mathcal U(R)>0\}
$$

It is easy to see that $\overline R(x)>0$. When $x$ converges to $0$, $\overline R(x)$
converges to the nonzero limit given by
$$
\overline R(0) = \max\{R>0: (1-\alpha)( \mathcal U(R+a) - k/\lambda)>
\mathcal U(R)\}. 
$$
With the aim of arriving at a contradiction, assume that there
exists a point $(x,\hat R)$ with $\hat R\leq \overline R(x)$ and
$V(x,\hat R)>MV(x,\hat R)$. In
other words, this point belongs to the consumption region. Let
$\check R = \max\{R<\hat R: V(x,R) = MV(x,R)\}$. In view of the
above remark, $\check R\geq 0$.  The points between $\check R$
and $\hat R$ belong to the consumption region, and therefore, the value
function satisfies the equation 
$$
\frac{\partial V(x,R)}{\partial R} = u_1(rV(x,R))
$$
on $(\check R, \hat R)$. Since $u_1$ is decreasing, for $R\in (\check R, \hat R)$,
\begin{eqnarray*}
u_1(rV(x,R)) & \leq & u_1(rV(x,\check R))\\
& = & u_1(rMV(x,\check R))\\ 
&\leq & u_1\left( r(1-e^{-\lambda{x}})
  \left( \mathcal U(\check R + a)   -
  \frac{k}{\lambda}\right) +re^{-\lambda{x}}
    \mathcal U(\check R)\right),
\end{eqnarray*}
and we have that 
$$
V(x,R) \leq MV(x,\check R) + (R - \check R) u_1\left( r(1-e^{-\lambda{x}})
  \left( \mathcal U(\check R + a)   -
  \frac{k}{\lambda}\right) +re^{-\lambda{x}}
    \mathcal U(\check R)\right).
$$
On the other hand, 
 since $V$ is increasing in $R$, 
{\footnotesize \begin{eqnarray*}
 MV(x,R) - MV(x,\check R) & = & \lambda\int_0^{x} e^{-\lambda h}
dh \left\{V(x-h,R+a) - V(x-h,\check R+a)\right\} +
e^{-\lambda{x}}( \mathcal U(R) -  \mathcal U(\check R))\\
&\geq & e^{-\lambda{x}}( \mathcal U(R) -  \mathcal U(\check R))
\end{eqnarray*}}
Combining the above estimates, passing to the limit $R\downarrow
\check R$, we get
{\footnotesize $$
\liminf_{R\downarrow \check R}\frac{MV(x, R) - V(x, R)}{ R - \check R} 
\geq   e^{-\lambda{x}} \mathcal U'(\check R)- u_1\left( r(1-e^{-\lambda{x}})
  \left( \mathcal U(\check R + a)   -
  \frac{k}{\lambda}\right) +re^{-\lambda{x}}
    \mathcal U(\check R)\right). 
$$}
On the other hand,
$$
 \mathcal U'(\check R) = u_1(r \mathcal U(\check R)), 
$$
and by definition of $\overline R(x)$, 
$$
e^{-\lambda{x}}u_1(r \mathcal U(\check R)) > u_1\left( r(1-e^{-\lambda{x}})
  \left( \mathcal U(\check R + a)   -
  \frac{k}{\lambda}\right) +re^{-\lambda{x}}
    \mathcal U(\check R)\right),
$$
which contradicts the assumption that $V(x,\hat R)>MV(x,\hat R)$. 

\textit{Step 2. Behavior near exhaustion.}
  By Proposition \ref{apriori.prop}, as $x\to0$, $V(R) = \mathcal U(R) + O(x)$ uniformly on compacts in $R$. Therefore,
  $$
  MV(x,R) = \mathcal U(R+a) \lambda x + \mathcal U(R) (1-\lambda x) - kx + O(x^2), 
  $$
  and
  $$
  MV'_R(x,R) = \mathcal U'(R+a) \lambda x + \mathcal U'(R) (1-\lambda x) + O(x^2),
  $$
  and the function $g$ of Lemma \ref{alter.lm} satisfies
{\footnotesize  $$
  g(x,R) = \mathcal U(R)^{\frac{1}{\alpha}-1} \mathcal U'(R) \, \left\{1 + x \frac{1-\alpha}{\alpha}\frac{\lambda(\mathcal U(R+a)-\mathcal U(R))-k}{\mathcal U(R)} + x\lambda \frac{\mathcal U'(R+a)-\mathcal U'(R)}{\mathcal U'(R)}\right\} + O(x^2). 
  $$}
  Since the exercise frontier is defined by $g(x,R^*(x)) = c$ with $c$ given in Lemma \ref{alter.lm}, or, equivalently,
  $$
   g(x,R^*(x)) = \mathcal U(R)^{\frac{1}{\alpha}-1} \mathcal U'(R),
   $$
   it follows that as $x\to 0$, the exercise frontier satisfies
   $$
   \frac{1-\alpha}{\alpha}\frac{\lambda(\mathcal U(R^*(x)+a)-\mathcal U(R^*(x)))-k}{\mathcal U(R^*(x))} + \lambda \frac{\mathcal U'(R^*(x)+a)-\mathcal U'(R^*(x))}{\mathcal U'(R^*(x))} = O(x). 
   $$
   Note that we could substitute $R^*(x)$ into the expansion for $g$ because the expansion is uniform on compacts and $R^*$ is bounded by Lemma \ref{uppercons.prop}.

\textit{Step 3. $R^*(x)$ is decreasing and the value function is increasing in $x$.}\   
Recall the alternative characterization of the exploration frontier in Lemma \ref{alter.lm}. Given the results of Lemma \ref{prop.lm}, to prove that $R^*$ is decreasing, it is enough to show that  
$$
\frac{\partial g(x,R)}{\partial x}\Big|_{R = R^*(x)}< 0. 
$$
Denoting by $\sim$ equality up to a multiplicative constant, this derivative satisfies:
\begin{eqnarray*}
\frac{\partial g(x,R^*(x))}{\partial x} & \sim & \frac{1-\alpha}{\alpha} V'_R(x,R^*(x))V^{\frac{1-2\alpha}{\alpha}}(x,R^*(x)) \frac{\partial MV(x,R^*(x))}{\partial x} \\ 
& & \qquad \qquad  + V^{\frac{1-\alpha}{\alpha}}(x,R^*(x))\frac{\partial MV'_R(x,R^*(x))}{\partial x}\\
& \sim & \frac{1-\alpha}{\alpha} V^{-1}(x,R^*(x)) \{-k +\lambda(V(x,R^*(x)+a) - V(x,R^*(x)))\} \\ 
& & \qquad \qquad  + V^{\frac{1-\alpha}{\alpha}}(x,R^*(x))\lambda(V^{\frac{\alpha-1}{\alpha}}(x,R^*(x)+a) - V^{\frac{\alpha-1}{\alpha}}(x,R^*(x)))\\
&\sim & \frac{1-\alpha}{\alpha} \left\{-\frac{k}{\lambda V(x,R^*(x))} -1 + \frac{V(x,R^*(x)+a)}{V(x,R^*(x))}\right\} \\ 
& & \qquad \qquad  - 1+ \left(\frac{V(x,R^*(x)+a))}{V(x,R^*(x))}\right)^{\frac{\alpha-1}{\alpha}}. 
\end{eqnarray*}
Using the explicit form of the value function in the consumption region, we can further write it as follows:
{\footnotesize 
$$
\frac{\partial g(x,R^*(x))}{\partial x}
\sim  -(1-\alpha) \frac{k}{\lambda V(x,R^*(x))} -1 + (1-\alpha)\left\{1+\frac{ca}{V(x,R^*(x))^{\frac{1}{\alpha}}}\right\}^{\alpha} 
     + \alpha\left\{1+\frac{ca}{V(x,R^*(x))^{\frac{1}{\alpha}}}\right\}^{\alpha-1}.
$$}
In other words,
$$
\frac{\partial g(x,R^*(x))}{\partial x}\sim \phi\left(\frac{(ca)^\alpha}{V(x,R^*(x))}\right)
$$
with
$$
\phi(y)  =  -(1-\alpha) \varepsilon y -1 + (1-\alpha)\left\{1+{y^{\frac{1}{\alpha}}}\right\}^{\alpha} 
     + \alpha\left\{1+{y^{\frac{1}{\alpha}}}\right\}^{\alpha-1}.
$$
where we recall that $\varepsilon = \frac{k}{(ca)^\alpha\lambda}<1$. 
Let us study the function $\phi$. It holds that $\phi(0)=0$, and its derivative is given by
$$
\phi'(y) = -(1-\alpha)\left\{\varepsilon - \frac{1}{\left(1+\frac{1}{ y^{1/\alpha}}\right)^{2-\alpha}}\right\}.
$$

The funciton $\phi'(y)$ is increasing, continuous, and satisfies $\phi'(0)= -(1-\alpha)\varepsilon<0$ and $\phi'(y)\to (1-\alpha)(1-\varepsilon)>0$ as $y\to \infty$. Therefore there exists a unique point $y^*>0$, depending only on $\varepsilon$, such $\phi(y)<0$ for $x\in (0,x^*)$ and $\phi(x)<0$ for $x>x^*$. 
This shows that $R^*(x)$ is strictly decreasing at all points $x$ such that $V(x,R^*(x))> V^*:=\frac{U(a)}{x^*}$, and strictly increasing at all points $x$ such that $V(x,R^*(x))< V^*$. 

Let $R^\dag(x)$ denote the unique solution of the equation $V(x,R^\dag(x))=V^*$. Since $V(x,R)$ is continuous in $x$,  it follows that $R^\dag$ is a continuous function of $x$. 
Moreover, by Step 2, 
$$
\phi\left(\frac{(ca)^\alpha}{\mathcal U(R_0)}\right) = 0,
$$
which means that $R^\dag(0) = R^*(0)$. To prove that $R^*$ is decreasing, it is thus sufficient to show that the curve $R^*(x)$ lays above the curve $R^\dag(x)$. Furthermore, let $V^\circ$ be the solution of 
$$
\left\{(V^\circ)^{\frac{1}{\alpha}} + ca\right\}^\alpha - V^\circ - \frac{k}{\lambda}=0,
$$
and let $R^\circ(x)$ be the unique solution of $V(x,R^\circ(x))=V^\circ$. It is easy to check that $V^\circ>V^*$ and hence $R^\circ(x)>R^\dag(x)$ for all $x$. In view of Equations  (\ref{dervx1}) and (\ref{dervx2}), to prove that $V$ is increasing in $x$, it is enough to show that the curve $R^*(x)$ lays below the curve $R^\circ(x)$.

We conclude the proof with the following geometric argument. 
\begin{itemize}
    \item $R^*(x)$ is below $R^\circ(x)$. Suppose that this is not the case. Then, there is a point $x_0$ such that $R^*(x_0)=R^\circ(x_0)$ and $\frac{\partial R^*(x_0)}{\partial x}\geq \frac{\partial R^\circ(x_0)}{\partial x}$. However, since $R^\circ(x_0)>R^\dag(x_0)$, $\frac{\partial R^*(x_0)}{\partial x}<0$, while by definition $\frac{\partial V(x,R^\circ(x))}{\partial x} = 0$ when $R^\circ(x)= R^*(x)$ and so $\frac{\partial R^\circ(x_0)}{\partial x}=0$ and we arrive to a contradiction. This shows that $V$ is increasing in $x$, and in particular that $R^\dag$ is decreasing. 
    \item $R^*(x)$ is above $R^\dag(x)$.  Assume otherwise, so that there exists $x_1>0$ such that $R^\dag(x_1)> R^*(x_1)$. Let $x_0 < x_1$  be the largest $x$ in on $[0,x_1]$ such that $R^\dag(x_0) \leq R^*(x_0)$. Then $R^\dag(x_0)=R^*(x_0)$. Since $R^\dag$ is decreasing, we must have $R^*(x_1)< R^*(x_0)$, which is not possible since $R^*$ is increasing on $[x_0,x_1]$. 
So  $R^\dag(x)\leq R^*(x)$ for all $x>0$,  which implies that $R^*$ is decreasing, as claimed.
\end{itemize}

\end{proof}

\subsection{Characterization of the value function and the optimal strategy}
We are now ready to prove the first and the third main theorems of this paper. 

\begin{proof}[Proof of Theorems \ref{HJB.thm} and \ref{strat.thm}]
Various properties of the value function have been shown in Proposition \ref{HJB.prop}, Lemma \ref{prop.lm} and the proof of Theorem \ref{cons.thm} above.

We now concentrate on the second part of Theorem \ref{HJB.thm} and on the characterization of the optimal strategy. 
Assume that the function $\widetilde V$ satisfies the assumptions of the theorem and let us show that it coincides
with the value function. Let $(\theta,c)\in \mathcal A(x,R)$ be an admissible bang-bang consumption-exploration strategy. Recall that the dynamics of the reserve process $(R_t)$ is given in equation (\ref{budget}), and the dynamics of the explored area process in equation \ref{area.eq}. 
Fix $T<\infty$. 
Then,
\begin{eqnarray*}
& & e^{-rT} \widetilde V(x-X_T,R_T) - \widetilde V(x,R)
= \int_0^T e^{-rt} (-r\widetilde V(x-X_t,R_t)-\frac{\partial
  \widetilde V}{\partial R} c_t) dt  \\
 & & + \sum_{n=0}^{\infty} \mathbf
  1_{\theta_n \leq T} e^{-r\theta_n}\{\widetilde V(x-X_{n+1},R_{\theta_n-}+a\mathbf 1_{X_{n+1}<x}) -\widetilde V(x-X_n,R_{\theta_n-})\}.
\end{eqnarray*}
Taking the expectation of both sides and using the fact that $\widetilde V$ satisfies the HJB equation, this implies that 
{\footnotesize \begin{eqnarray*}
\widetilde V(x,R) &\geq & \mathbb E\left[\int_{0}^{T}
       e^{-rt} u(c_t)dt\right] + 
\mathbb  E[e^{-rT}\widetilde V(x-X_T,R_T)]\label{ineq1}\\
& & \qquad-  \mathbb E\left[\sum_{n=0}^{\infty} \mathbf
  1_{\theta_n \leq T} e^{-r\theta_n}\mathbb E\{\widetilde V(x-X_{n+1},R_{\theta_n-}+a\mathbf 1_{X_{n+1}<x}) -\widetilde V(x-X_n,R_{\theta_n-})|\mathcal F_{X_n}\}\right]\nonumber\\
& = &  \mathbb E\left[\int_{0}^{T}
       e^{-rt} u(c_t)dt\right] +  \mathbb E[e^{-rT}\widetilde V(x-X_T,R_T)] \nonumber\\
& & \qquad-  \mathbb E\left[\sum_{n=0}^{\infty} \mathbf
  1_{\theta_n \leq T} e^{-r\theta_n}\{M\widetilde V(x-X_n,R_{\theta_n-})-\widetilde V(x-X_n,R_{\theta_n-}) + k\frac{1-e^{-\lambda X_n}}{\lambda}\}\right]\nonumber\\
&\geq & \mathbb E\left[\int_{0}^{T}
       e^{-rt} u(c_t)dt - k\sum_{n=0}^{\infty} e^{-r\theta_n}\mathbf
  1_{\theta_n \leq T} (X_{n+1}-X_n) \right] +
 \mathbb E[e^{-rT}\widetilde V(x-X_T,R_T)], \label{ineq2}
\end{eqnarray*}}
where we have used the fact that $\theta_n\in \mathcal F_{X_n}$ and
$X_{n+1}-X_n$ is independent from $\mathcal F_{X_n}$.  As $T\to \infty$,
\begin{eqnarray*}
0 & \leq & \lim_{T\to \infty}  \mathbb E[e^{-rT}\widetilde V(x-X_T,R_T)] \\
  & \leq & \lim_{T\to \infty}  C\mathbb E[e^{-rT}(1+ U(R_T))] \leq \lim_{T\to \infty}  C\mathbb E[e^{-rT}(1+ U(R+aN_{x}))] =0,
\end{eqnarray*}
so that 
$$
\widetilde V(x,R) \geq \mathbb E\left[\int_{0}^{\infty}
       e^{-rt} u(c_t)dt - k\sum_{n=0}^{\infty} e^{-r\theta_n}
  (X_{n+1}-X_n)\right] ,
$$
and since the consumption-exploration strategy was arbitrary,
$$
\widetilde V(x,R)  \geq V(x,R). 
$$

To prove the equality, remark that Lemma \ref{cinf} and Proposition \ref{boundary.prop} hold with $\widetilde V$ instead of $V$, which means that we can define a consumption-exploration strategy using formulas (\ref{strat.c0}--\ref{strat.Rn}) with $\widetilde V$ instead of $V$ and with $\widetilde R^*$, the exploration frontier defined from the function $\widetilde V$ instead of $R^*$. With this choice of consumption-exploration strategy, for all $n\geq 0$,  $R_{\theta_n-}\leq R^*(x-X_n)$ with the convention $\theta_{-1}=0$, so that 
$$
M\widetilde V(x-X_n,R_{\theta_n-})   =\widetilde  V(x-X_n,R_{\theta_n-}),
$$
This means that the inequality in (\ref{ineq2}) becomes an equality. On the other hand, for $R\geq R^*(x)$, the function $\widetilde V$ satisfies the equation $u^*\left(\frac{\partial \widetilde V}{\partial R}(x,R)\right) = r \widetilde V(x,R)$. From this, it is easy to deduce that
$$
c^n(t) = \left(\frac{\partial \widetilde V}{\partial R}(x-X_n, R^n(t))\right)^{\frac{1}{\alpha-1}},
$$
and so
$$
r\widetilde V (x-X_n, R^n(t)) + c^n(t)\frac{\partial \widetilde V}{\partial R}(x-X_n, R^n(t)) = e^{-rt} u(c^n(t)). 
$$
This means that the inequality in (\ref{ineq1}) also becomes an equality and we conclude that 
$$
\widetilde V(x,R) = V(x,R).
$$
This argument also establishes the optimality of the strategy (\ref{strat.c0}--\ref{strat.Rn}) and thus proves Theorem \ref{strat.thm}.
\end{proof}

\subsection{Optimality of bang-bang strategies}

\begin{proof}[Proof of Theorem \ref{optimality}]

Consider an admissible consumption-exploration strategy $(c,X)$ and its associated reserve process $R$, and let $t>0$. 
Since $V$ is continuously differentiable in its entire domain, 
\begin{eqnarray*}
e^{-rt} V(x-X_t,R_t) - V(x,R) &=& \int_0^t e^{-rs}\left\{-rV   - \frac{\partial V}{\partial R} c_s \right\}ds-\int_0^t e^{-rs} \frac{\partial V}{\partial x} dX^c_s \\&+& \sum_{0\leq s\leq t: \Delta X_{s} \neq 0\text{\ or\ }\Delta R_s\neq 0} e^{-r s}\left\{V(x-X_s,R_{s}) - V(x-X_{s-},R_{s-})\right\},
\end{eqnarray*}
where $X^c$ denotes the continuous part of $X$. 

Introduce the function $\theta(u) = X_{X^{-1}(u)}$, the process $\overline R_u = R_{X^{-1}(u)} + a(N_u-N_{\theta(u)})$ and the function
$$
F(y) = \int_0^{y}e^{-r X^{-1}(u)}\frac{\partial V(x-u,\overline R_u)}{\partial x} du. 
$$
Since
$$
F(X_t) = \int_0^t F'(X_s) dX^c_s + \sum_{0\leq s\leq t: \Delta X_s \neq 0}(F(X_s)-F(X_{s-})),
$$
it follows that 
\begin{eqnarray*}
\int_0^t e^{-rs} \frac{\partial V(x-X_s,R_s)}{\partial x} dX^c_s&=&\int_0^{X_t} e^{-rX^{-1}(u)}\frac{\partial V(x-u,\overline R_u)}{\partial x} du \\&-& \sum_{0\leq s\leq t: \Delta X_s \neq 0} e^{-rs}\int_{X_{s-}}^{X_s} \frac{\partial V}{\partial x} (x-u,\overline R_u) du.
\end{eqnarray*}
The last term equals
\begin{eqnarray*}
&&\sum_{0\leq s\leq t: \Delta X_s \neq 0} e^{-rs}\int_{X_{s-}}^{X_s} \frac{\partial V}{\partial x} (x-u,\overline R_{X_{s-}} + a(N_u - N_{X_{s-}})) du
\\
&&\qquad \qquad =- \sum_{0\leq s\leq t: \Delta X_s \neq 0} e^{-rs} \int_{X_{s-}}^{X_s} \{V(x-u,\overline R_{u-}+a) - V(x-u,\overline R_{u-}) \} dN_u \\ &&\qquad\qquad\qquad\qquad + \sum_{0\leq s\leq t: \Delta X_s \neq 0} e^{-rs}\left\{ V(x-X_s,R_s) - V(x-X_{s-},R_{s-})\right\}\\
&&\qquad\qquad =  -\int_0^{X_t} e^{-rX^{-1}(u)}  \{V(x-u,\overline R_{u-}+a) - V(x-u,\overline R_{u-}) \} dN_u \\&&\qquad\qquad\qquad\qquad+ \sum_{0\leq s\leq t: \Delta X_s \neq 0\text{\ or\ } \Delta R_s \neq 0} \{V(x-X_s,R_s) - V(x-X_{s-},R_{s-})\}. 
\end{eqnarray*}
Thus,
{\footnotesize\begin{eqnarray*}
&&e^{-rt} V(x-X_t,R_t) - V(x,R) = \int_0^t e^{-rs}\left\{-rV  - \frac{\partial V}{\partial R} c_s \right\}ds\\ &&-\int_0^{X_t} e^{-rX^{-1}(u)}\frac{\partial V(x-u,\overline R_u)}{\partial x} du +\int_0^{X_t} e^{-rX^{-1}(u)}  \{V(x-u,\overline R_{u-}+a) - V(x-u,\overline R_{u-}) \} dN_u.
\end{eqnarray*}}
Applying the inequality (\ref{dervx2}), we then find:
\begin{eqnarray*}
e^{-rt} V(x-X_t,R_t) &-& V(x,R) \leq -\int_0^t e^{-rs}u(c_s)ds + k\int_0^t e^{-rs}d X_s  \\&+& \int_0^{X_t} e^{-rX^{-1}(s)} \left\{V(x-u,\overline R_{u-} + a) - V(x-u,\overline R_{u-})\right\}(dN_u - \lambda du).
\end{eqnarray*}
Finally, taking the expectation and using the martingale property, we conclude:
$$
V(x,R)\geq \mathbb E[e^{-rt} V(x-X_t,R_t)] + \mathbb E\left[\int_0^t e^{-rs}u(c_s)ds - k\int_0^t e^{-rs}d X_s \right].
$$
Passing to the limit $t\to \infty$ and using the positivity of the value function, the proof is completed. 
\end{proof}
\end{document}